\documentclass[preprint]{imsart}

\RequirePackage[numbers]{natbib}
\RequirePackage[colorlinks,citecolor=blue,urlcolor=blue]{hyperref}

\usepackage{amsmath,amsthm,amsfonts,epsfig,amssymb,natbib,eucal, graphicx, float, multirow, booktabs, algorithm, pdflscape}
\usepackage[noend]{algpseudocode}

\newtheoremstyle{break}
  {\topsep}{\topsep}%
  {\itshape}{}%
  {\bfseries}{}%
  {\newline}{}%
\theoremstyle{break}

\newtheorem{lemma}{Lemma}
\newtheorem{theorem}{Theorem}

\newtheorem{definition}{Definition}

\newtheorem{result}{Result}

\arxiv{}

\startlocaldefs

\DeclareMathOperator*{\Vphi}{\boldsymbol{\varphi}}

\DeclareMathOperator*{\btheta}{\boldsymbol{\theta}}

\endlocaldefs

\begin{document}

\begin{frontmatter}

\title{Efficient methods for the estimation of the multinomial parameter for the two-trait group testing model}
\runtitle{Multinomial group testing estimation}


\author{\fnms{Gregory} \snm{Haber}\ead[label=e1]{habergw@nih.gov}}
\address{Biostatistics Branch, Division of Cancer Epidemiology and Genetics, National Cancer Institute, NIH,
  Bethesda, MD 20892, USA \printead{e1}}
\and
\author{\fnms{Yaakov} \snm{Malinovsky}\ead[label=e2]{yaakovm@umbc.edu}}
\address{Department of Mathematics and Statistics, University of Maryland, Baltimore County, Baltimore, MD 21250,
  USA \printead{e2}}

\runauthor{G. Haber and Y. Malinovsky}

\begin{abstract}
  Estimation of a single Bernoulli parameter using pooled sampling is among the oldest problems in the group testing literature. To carry out such estimation, an array of efficient estimators have been introduced covering a wide range of situations routinely encountered in applications. More recently, there has been growing interest in using group testing to simultaneously estimate the joint probabilities of two correlated traits using a multinomial model. Unfortunately, basic estimation results, such as the maximum likelihood estimator (MLE), have not been adequately addressed in the literature for such cases. In this paper, we show that finding the MLE for this problem is equivalent to maximizing a multinomial likelihood with a restricted parameter space. A solution using the EM algorithm is presented which is guaranteed to converge to the global maximizer, even on the boundary of the parameter space. Two additional closed form estimators are presented with the goal of minimizing the bias and/or mean square error. The methods are illustrated by considering an application to the joint estimation of transmission prevalence for two strains of the Potato virus Y by the aphid {\em Myzus persicae}.
\end{abstract}





\end{frontmatter}

\section{Introduction}
Estimation of some trait in a population when the (unknown) prevalence is rare
and/or only a limited number of tests can be performed is a difficult statistical problem. One approach in such cases is group testing, in which
individuals are screened in pools as opposed to individually. Depending on the
underlying prevalence and group size, such methods have been shown to yield large gains in efficiency (as measured by mean square error (MSE)) and,
often, a reduction in the total number of tests required \citep[see, for example,][]{thompson1962,swallow1985}. Applications can be found in a wide range of areas, although uses in plant and animal sciences are especially common.

The standard group testing problem, in which the goal is estimation of a single trait, has been well studied, yielding an array of efficient estimators \citep[see, as a few examples,][]{burrows1987, tebbs2003, hepworth2009, santos2016, haber2018}. More recently, however, there has been interest in the use of pooled sampling for the simultaneous estimation of two or more traits \citep[see, for example,][]{ho2000, pfeiffer2002, ding2015, tebbs2013, warasi2016, li2017, hyun2018}. This interest has been spurred by the growing availability of multiplex assays in many areas of science designed for screening multiple diseases simultaneously. The benefits of such assays are clear, allowing for reductions in the number of tests needed for gathering data on two or more traits, as well as providing information on the joint distribution of the characteristics under study. This has led to the need to develop new statistical methods to handle data generated from such tools and to efficiently extract information on the underlying multivariate distribution. Furthermore, new methods for designing pooling studies utilizing multiplex assays which take the correlations among diseases into account have been needed as well. While research in these areas has been ongoing, unfortunately, even in the case of two- diseases, results for small sample estimation do not exist, and basic tools such as maximum likelihood estimation have not been adequately addressed.

Conceptually, finding the maximum likelihood estimator (MLE) for the two-trait group testing problem can be expressed as a special case of maximizing a multinomial likelihood with a constrained parameter space (details are provided in the following section). Due to the restrictions on the parameter space, closed form techniques, such as those based on the invariance property of the MLE, do not work in this case \citep[a recent work,][did report a closed form MLE based on this principle, but it can easily be checked that it yields estimates outside the parameter space]{li2017}. While numerical methods are possible, the restriction means that many estimates will fall on the boundary of the parameter space, and it is difficult to ensure convergence to a global maximum in such cases.

Maximization under constrained parameter spaces is a well studied problem for a variety of statistical models \citep[see, for example,][]{net1999, jam2004}. Numerical methods specifically for the multinomial model under a range of convex constraints have been developed as well \citep[see][and the references therein]{grend2017}.

While these previous methods can be adapted to the group testing problem, our goal here is to provide a much simpler solution for this special case which is guaranteed to yield a global maximizer. We show that, when optimizing over the boundary, the problem is equivalent to a convex optimization problem in one fewer dimension. The maximization can then be carried out using a variety of methods, and we develop here an EM algorithm-based approach. The resulting estimator is shown to converge to the unique global maximizer.

Two additional closed form estimators are presented as well. One, based on the method of moments, is shown to approximate the MLE very closely. The second, a shrinkage estimator based on the one-trait group testing estimator presented in \citet{burrows1987}, is developed with the intention of reducing the MSE and/or bias of the MLE.

Numerical comparisons in terms of relative bias and MSE are presented which cover a wide range of applicable situations. The methods are further illustrated by considering two experiments where the transmission rates for prevalences of different strains of Potato virus Y are to be estimated simultaneously.

\section{Statistical model}
Let $\varphi_{1}$ and $\varphi_{2}$ be marginally Bernoulli random variables with
parameters $0 < p_1 < 1$ and $0 < p_2 < 1$, respectively, each indicating the presence of a given trait.
 Then, $(\varphi_{1}, \varphi_{2})$ has a one-to-one correspondence to the vector $\Vphi =
(\varphi_{10}, \varphi_{01}, \varphi_{11})$ with joint multinomial
distribution $\Vphi \sim MN_3(1, \mathbf{p})$ and parameter space
$\boldsymbol{\Psi}_{\mathbf{p}} = \{\mathbf{p}: \mathbf{1}^\prime \mathbf{p} <
1, 0 \prec \mathbf{p} \prec 1\}$, where $\mathbf{1} = (1, 1, 1)^\prime$, $\mathbf{p} = (p_{10}, p_{01}, p_{11})$
and $p_{00} = 1 - \mathbf{1}^\prime \mathbf{p} = 1 - p_{10} - p_{01} - p_{11}$ and $\prec$ denotes element-wise
inequality. Note that the marginal parameters can be expressed as $p_1 = p_{10} + p_{11}$ and $p_2 = p_{01} + p_{11}$.
Throughout this work, our primary interest is estimation of the parameter $\mathbf{p}$.

The $i^{th}$ pooled sample comprised of $k$ individual units can then be represented by the random variable $(\vartheta^{i}_{1}, \vartheta^{i}_{2}) =(\max\{\varphi_{1_1}, \cdots, \varphi_{1_k}\}, \max\{\varphi_{2_1}, \cdots, \varphi_{2_k}\})$ which corresponds to 
\[\boldsymbol{\vartheta}^i = (\vartheta^i_{10}, \vartheta^i_{01}, \vartheta^i_{11}) \sim MN_3(1, \btheta),\]
where 
\begin{align}
\label{eq:theta14}
\btheta &= (\theta_{10}, \theta_{01}, \theta_{11}) \nonumber \\
&= ((p_{00} + p_{10})^k - p_{00}^k, (p_{00} + p_{01})^k - p_{00}^k,1 - (p_{00} + p_{10})^k - (p_{00} + p_{01})^k + p_{00}^k)
\end{align} and 
\begin{equation}
\label{eq:theta24}
\theta_{00} = 1 - \mathbf{1}^\prime \btheta = p_{00}^k. 
\end{equation}

If we sample $n$ such groups, we have the random variable $\mathbf{x} = (x_{10}, x_{01}, x_{11}) = \sum_{i=1}^n \boldsymbol{\vartheta}^i \sim
MN_3(n, \btheta)$ with parameter space $\boldsymbol{\Psi}_p =
\{\btheta(\mathbf{p}): \mathbf{1}^\prime \mathbf{p} < 1, 0 \prec \mathbf{p}
\prec 1\}$. For later use, we define $x_{00} = n - (x_{10} + x_{01} + x_{11})$.

It should be noted that $\boldsymbol{\Psi}_p$ is a proper subset of
the full parameter space $\boldsymbol{\Psi}_{\btheta} = \{\btheta:
\mathbf{1}^\prime \btheta < 1, 0 \prec \btheta \prec 1\}$. For example, with $k
= 2$, $\btheta = (0.45, 0.45, 0.05) \in \boldsymbol{\Psi}_{\btheta}$ is achieved if and only if
$\mathbf{p} = (0.484, 0.484, -0.192) \notin \boldsymbol{\Psi}_p$. As
such, maximizing the likelihood with respect to $\mathbf{p} \in
\boldsymbol{\Psi}_p$ is equivalent to the problem of maximizing a standard
multinomial likelihood with respect to $\btheta$ such that the estimate lies in
the restricted parameter space $\boldsymbol{\Psi}_p$.

For use in later results, we define the closure $\overline{\boldsymbol{\Psi}}_p
= \boldsymbol{\Psi}_p \cup \partial \boldsymbol{\Psi}_p$ where $\partial
\boldsymbol{\Psi}_p$ is the boundary of the parameter space. Likewise, let
$\mathcal{X} = \{\mathbf{x}: 0 \preceq \mathbf{x} \preceq n\}$, where $\preceq$ denotes element wise non-strict inequality, be the sample space of $\mathbf{x}$ with interior $\mathcal{X}_0 = \{\mathbf{x}: 0 \prec \mathbf{x} \prec n\}$.

\section{Maximum likelihood estimation}
We seek to maximize the log-likelihood
\begin{equation}
\label{eq:ll}
\ell(\btheta|\mathbf{x}) \propto x_{00} \log(\theta_{00}) + x_{10} \log(\theta_{10}) + x_{01}
\log(\theta_{01}) + x_{11} \log(\theta_{11}),
\end{equation} such that $\btheta \in \boldsymbol{\Psi}_p$.

The following lemma, the proof of which is given, together with all subsequent proofs, in Appendix \ref{sec:ap}, establishes the concavity of the log-likelihood function.
\begin{lemma}~\vspace{-15pt}
	\label{lm:llcon}
	\begin{itemize}
		\item[(a)] For $\mathbf{x} \in \mathcal{X}_0$, the log-likelihood given in (\ref{eq:ll}) is strictly
		concave for all $\mathbf{p} \in \boldsymbol{\Psi}_\theta$.
		\item[(b)] For all $\mathbf{x} \in \mathcal{X}$, the log-likelihood given in (\ref{eq:ll}) is
		concave (not necessarily strict) for all $\mathbf{p} \in
		\boldsymbol{\Psi}_\theta$.
	\end{itemize}
\end{lemma}

Using standard multinomial theory, we know that, when maximizing over
$\boldsymbol{\Psi}_{\btheta}$, the unique MLE is given by $\hat{\btheta}^{MLE} =
\bar{\mathbf{x}} = \frac{\mathbf{x}}{n}$ for $\mathbf{x} \in \mathcal{X}_0$, and this point is the
unique maximizer of the likelihood for all $\mathbf{x} \in \mathcal{X}$ over the
closure of $\boldsymbol{\Psi}_{\btheta}$.
To find the MLE under the restricted parameter space, $\boldsymbol{\Psi}_p$, we
first note that there exists a one-to-one mapping $\btheta \mapsto \mathbf{p}$ as given in the following lemma. The proof of this lemma is found by inverting (\ref{eq:theta14}) and (\ref{eq:theta24}).

\begin{lemma}
	\label{lm:h24}
	The unique function $h:\btheta \mapsto \mathbf{p}$ is given by
	\begin{align*}
	p_{00} &= h_{00}(\btheta) = (1 - \theta_{10} - \theta_{01} -
	\theta_{11})^{1 / k},\\
	p_{10} &= h_{10}(\btheta) = (1 - \theta_{01} - \theta_{11})^{1 / k} -
	h_{00}(\btheta),\\
	p_{01} &= h_{01}(\btheta) = (1 - \theta_{10} - \theta_{11})^{1 / k} -
	h_{00}(\btheta),\\
	p_{11} &= h_{11}(\btheta) = 1 - p_{00} - p_{10} - p_{01}.	
	\end{align*}
	
\end{lemma}

If we define the set \[R_n = \left\{\mathbf{x}: \left(\frac{x_{00} +
	x_{10}}{n}\right)^{1 /k} + \left(\frac{x_{00} + x_{01}}{n}\right)^{1 /k}
- \left(\frac{x_{00}}{n}\right)^{1 /k} < 1\right\}\] then, for any
$\mathbf{x} \in  \mathcal{X}_0 \cap R_n$, $h(\bar{\mathbf{x}}) \in
\boldsymbol{\Psi}_p$ where $h$ is as in Lemma \ref{lm:h24}, so that such values
provide the unique MLE for $\mathbf{p}$ by the invariance property of the MLE. Since the log-likelihood is a concave function on $\boldsymbol{\Psi}_{\btheta}$, it is clear that the maximizer for all $\mathbf{x} \notin R_n$ must lie in $\partial \boldsymbol{\Psi}_p$. This leads to the following result.

\begin{theorem}[Existence and uniqueness of MLE]~\vspace{-15pt}
	\label{md:mle1}
	\begin{itemize}
		\item[(a)] A necessary and sufficient condition for the maximum likelihood estimator of $\mathbf{p} \in \boldsymbol{\Psi}_p$ to exist and be unique is that $\mathbf{x} \in \mathcal{X}_0 \cap R_n$. In this case, the MLE is given by
		\begin{equation}
		\label{md:mle1eq}
		\hat{\mathbf{p}}^{MLE} = (\hat{p}_{10}^{MLE}, \hat{p}_{01}^{MLE}, \hat{p}_{11}^{MLE}), 
		\end{equation} where
		\begin{align*}
		\hat{p}_{00}^{MLE} &= \left(\frac{x_{00}}{n}\right)^{1 /k}, \\ 
		\hat{p}_{10}^{MLE} &= \left(\frac{x_{00} + x_{10}}{n}\right)^{1 /k} - \left(\frac{x_{00}}{n}\right)^{1 /k},\\
		\hat{p}_{01}^{MLE} &= \left(\frac{x_{00} + x_{01}}{n}\right)^{1 /k} - \left(\frac{x_{00}}{n}\right)^{1 /k},\\ \intertext{and}
		\hat{p}_{11}^{MLE} &= 1 - \hat{p}_{00}^{MLE} - \hat{p}_{10}^{MLE} - \hat{p}_{01}^{MLE}.
		\end{align*}
		\item[(b)] As $n\to \infty$, $P(\mathbf{x} \in \mathcal{X}_0 \cap R_n) = 1$ so that  the MLE as given in (\ref{md:mle1eq}) exists and is unique with probability one.
		
	\end{itemize}
\end{theorem}

\subsection{Maximization over the boundary}
While the above result is complete for large samples, in many cases $n$ will
necessarily be small and we will be interested in maximizing over the closure
$\overline{\boldsymbol{\Psi}}_p$. From here on,
$\hat{\mathbf{p}}^{MLE}$ will refer to the maximizer over this
extended space. Defining \[\overline{R}_n = \left\{\mathbf{x}: \left(\frac{x_{00} +
	x_{10}}{n}\right)^{1 /k} + \left(\frac{x_{00} + x_{01}}{n}\right)^{1 /k}
- \left(\frac{x_{00}}{n}\right)^{1 /k} \leq 1\right\},\] the invariance
property of the MLE and Lemma \ref{lm:h24} are sufficient to establish $\hat{\mathbf{p}}^{MLE} =
h(\bar{\mathbf{x}})$ for all $\mathbf{x} \in \overline{R}_n$.

To get an idea of how common boundary estimates can be, Table \ref{tab:prob1}
provides values of $P(\mathbf{x} \notin \overline{R}_n)$ for a wide range of $n$
with several realistic values of $\mathbf{p}$ and $k$. Note that, when $n = 1$, this probability is theoretically $0$ for all values of $k$ and $\mathbf{p}$.

\begin{table}[!htb]
	\caption{Values of $P(\mathbf{x} \notin \overline{R}_n)$ for varying $\mathbf{p}, n,$ and $k$.} 
	\label{tab:prob1}
	\centering
	\resizebox{\columnwidth}{!}{%
		\begin{tabular}{cccccc}
			\toprule
			& $(p_{10}, p_{01}, p_{11})=$ & $(0.045, 0.045, 0.005)$ & $(0.095, 0.045, 0.005)$ & $(0.1, 0.1, 0.1)$ & $(0.25, 0.05, 0.15)$ \\ \midrule
			& $n$ \\ \midrule
			$k = 2$& $5$ & 0.1029 & 0.1724 & 0.1376 & 0.0924 \\
			& $10$       & 0.2872 & 0.3935 & 0.0834 & 0.0380 \\
			& $15$       & 0.4299 & 0.5082 & 0.0349 & 0.0144 \\
			& $25$       & 0.5555 & 0.5465 & 0.0059 & 0.0015 \\
			& $50$       & 0.4819 & 0.4003 & 0.0002 & 0.0000 \\     
			& $100$      & 0.2475 & 0.2500 & 0.0000 & 0.0000 \\     
			& $500$      & 0.0194 & 0.0474 & 0.0000 & 0.0000 \\      
			& $1000$     & 0.0014 & 0.0081 & 0.0000 & 0.0000 \\       \midrule
			$k = 5$& $5$ & 0.2879 & 0.3540 & 0.2182 & 0.0849 \\
			& $10$       & 0.4418 & 0.4408 & 0.1575 & 0.1472 \\
			& $15$       & 0.4294 & 0.4288 & 0.0903 & 0.1663 \\
			& $25$      & 0.3666 & 0.3939 & 0.0314 & 0.1454 \\  
			& $50$      & 0.2826 & 0.3359 & 0.0028 & 0.0596 \\       
			& $100$     & 0.1894 & 0.2663 & 0.0000 & 0.0059 \\        
			& $500$     & 0.0196 & 0.0762 & 0.0000 & 0.0000 \\        
			& $1000$    & 0.0017 & 0.0211 & 0.0000 & 0.0000 \\        \midrule
			$k = 10$& $5$& 0.3969 & 0.3838 & 0.0909 & 0.0062 \\   
			& $10$      & 0.3984 & 0.4410 & 0.2367 & 0.0214 \\  
			& $15$      & 0.3989 & 0.4414 & 0.3328 & 0.0392 \\   
			& $25$      & 0.3419 & 0.3949 & 0.3783 & 0.0738 \\   
			& $50$      & 0.2883 & 0.3651 & 0.2330 & 0.1419 \\        
			& $100$     & 0.2085 & 0.3094 & 0.0593 & 0.2307 \\        
			& $500$     & 0.0328 & 0.1283 & 0.0000 & 0.2332 \\       
			& $1000$    & 0.0046 & 0.0542 & 0.0000 & 0.0768 \\       \midrule
			$k = 25$& $5$& 0.3477 & 0.1579 & 0.0003 & 0.0000 \\   
			& $10$      & 0.4515 & 0.3293 & 0.0012 & 0.0000 \\  
			& $15$      & 0.4277 & 0.4208 & 0.0027 & 0.0000 \\  
			& $25$      & 0.4301 & 0.4841 & 0.0073 & 0.0000 \\  
			& $50$      & 0.3650 & 0.4718 & 0.0272 & 0.0000 \\       
			& $100$     & 0.3070 & 0.4341 & 0.0917 & 0.0001 \\       
			& $500$     & 0.1204 & 0.3089 & 0.6579 & 0.0011 \\      
			& $1000$    & 0.0477 & 0.2326 & 0.8297 & 0.0025 \\     
			\bottomrule
	\end{tabular}}
\end{table}

From the table, we see that, when $p_{00}$ is large this probability can be
quite substantial, even for large values of $n$ and small $k$. While this effect
seems to lessen as the probability of at least one positive trait increases,
group testing is most commonly used in the context of rare traits.
It is apparent, then, that the problem of boundary estimates will be present in many applications.

Unfortunately, despite the concavity of the log-likelihood, the maximum on the
boundary will not occur at a stationary point. Furthermore, the problem as
previously expressed is not guaranteed to have a unique maximizer over the
boundary. As such, maximizing the likelihood over $\partial \boldsymbol{\Psi}_p$ is a non-trivial optimization problem.

To proceed, the following theorem allows us to reduce the dimension of the parameter space by one, facilitating the use of convex theory results to find the maximizer. 

\begin{theorem}
	\label{md:mle11}
	For all $\mathbf{x} \notin \overline{R}_n$ the log-likelihood given in (\ref{eq:ll}) over
	$\overline{\boldsymbol{\Psi}}_p$ is maximized at a point such that
	$\hat{p}_{11} = 0$. If $\mathbf{x} \in \mathcal{X}_0$ or $x_{11} = 0$ is the only zero, then
	the log-likelihood is uniquely maximized at a point with $\hat{p}_{11} = 0$.
\end{theorem}

As a result of Theorem \ref{md:mle11}, for values $\mathbf{x} \notin \overline{R}_n$ the
objective function given in (\ref{eq:ll})
can be expressed in terms of the simpler two-parameter model in which we seek
to maximize
\begin{equation}
\label{eq:ll2}
\begin{split}
\ell^*(\mathbf{p}^*|\mathbf{x}) &\propto x_{00} \log((1 - p_{10} - p_{01})^k) + x_{10} \log((1 - p_{01})^k
- (1 - p_{10} - p_{01})^k)\\ &+ x_{01} \log((1 - p_{10})^k
- (1 - p_{10} - p_{01})^k)\\ &+ x_{11} \log(1 - (1 - p_{01})^k - (1 - p_{10})^k
+ (1 - p_{10} - p_{01})^k)
\end{split}
\end{equation}
over the set $ \boldsymbol{\Psi}_{p^*} = \{\mathbf{p}^*: \mathbf{1}^\prime \mathbf{p}^*
< 1, 0 \prec \mathbf{p}^* \prec 1\}$ where $\mathbf{p}^* = (p_{10}, p_{01})$.

The following lemma addresses the concavity of this likelihood function. 
\begin{lemma}~\vspace{-15pt}
	\label{lm:llcon2}
	\begin{itemize}
		\item[(a)] For $\mathbf{x} \in \mathcal{X}_0 \cap \overline{R}_n^c$, the
		log-likelihood given in (\ref{eq:ll2}) is strictly
		concave for all $(\mathbf{p}^*, 0) \in \boldsymbol{\Psi}_{\theta}$.
		\item[(b)] For all $\mathbf{x} \in \mathcal{X} \cap \overline{R}_n^c$, the
		log-likelihood given in (\ref{eq:ll2}) is
		concave (not necessarily strict) for all $(\mathbf{p}^*, 0) \in
		\boldsymbol{\Psi}_{\theta}$.
	\end{itemize}
\end{lemma}

As a result of Lemma \ref{lm:llcon2} we establish that, for any $\mathbf{x} \notin \overline{R}_n$ maximization of (\ref{eq:ll2}) over $\boldsymbol{\Psi}_{p^*}$ can be carried out in a
number of ways with the resulting estimate, combined with $\hat{p}_{11} = 0$, yielding a global maximum of
(\ref{eq:ll}) which is the desired MLE. In the following section, we develop an EM algorithm-based approach to
solving this simplified optimization problem.

\subsection{EM algorithm}
In this section we derive an EM algorithm-based estimator assuming values $\mathbf{x} \notin \overline{R}_n$, for which we know the maximizing value takes $p_{11} = 0$. For the complete
data, we use the true underlying status of each individual in the study.

\begin{result}
  \label{res:EM}
	Beginning with an initial value $\mathbf{p}^{*(0)}$, the estimate from the $t^{th}$ iteration, $t = 1, 2, 3, \ldots$, of the EM algorithm is given by 
	\begin{align}
	\label{eq:em1}
	p_{10}^{(t)} &= \frac{(p_{00}^{(t-1)} + p_{10}^{(t-1)})^{k-1}p_{10}^{(t-1)}}{\theta_{10}^{(t-1)}}\frac{x_{10}}{n} + \frac{[1 - (p_{00}^{(t-1)} + p_{10}^{(t-1)})^{k-1}]p_{10}^{(t-1)}}{\theta_{11}^{(t-1)}}\frac{x_{11}}{n}\\
	p_{01}^{(t)} &= \frac{(p_{00}^{(t-1)} + p_{01}^{(t-1)})^{k-1}p_{01}^{(t-1)}}{\theta_{01}^{(t-1)}}\frac{x_{01}}{n} + \frac{[1 - (p_{00}^{(t-1)} + p_{01}^{(t-1)})^{k-1}]p_{01}^{(t-1)}}{\theta_{11}^{(t-1)}}\frac{x_{11}}{n} \\
	p_{00}^{(t)} &= 1 - p_{10}^{(t)} - p_{01}^{(t)}  \label{eq:em2}
	\end{align}
	
\end{result}

\subsection{Global maximum over the closure}
The previous results can be combined to yield an estimator which gives a global maximizer of the likelihood over
the closure, $\overline{\boldsymbol{\Psi}}_p$. The steps for finding this estimator are given in Algorithm \ref{alg}.

\begin{minipage}[c]{.95\textwidth}
\renewcommand*\footnoterule{}
\begin{algorithm}[H]
	\caption{Global Maximizer}
	\label{alg}
	\begin{algorithmic}[1]
		\If{$\mathbf{x} \in \overline{R}_n$}
		
		\State take $\hat{\mathbf{p}}^{MLE} = h(\mathbf{\bar{x}}),$ where $h(\cdot)$ is as in Lemma \ref{lm:h24}.
		\Else
		
		\State beginning with an initial value $\mathbf{p}^{*{(0)}}$, iterate $\mathbf{p}^{*(t)}, t = 1, 2, 3,
		\ldots$ as in (\ref{eq:em1}) - (\ref{eq:em2})  until convergence\footnote{\footnotesize {In this article, we use a likelihood based convergence criteria (e.g., stop when $|\ell(\mathbf{p}^{*{(t + 1)}}|\mathbf{x}) - \ell(\mathbf{p}^{*{(t)}}|\mathbf{x})| < \epsilon$ for some $\epsilon > 0$), but other criteria may be used as well.} \rule{\textwidth}{.6pt}} of $\ell$ (call the value at convergence
		$\mathbf{p}^{*(\infty)})$;
		
		\State take $\hat{\mathbf{p}}^{MLE} = (\mathbf{p}^{*(\infty)}, 0)$.
		\EndIf
	\end{algorithmic}
\end{algorithm}
\end{minipage}

For $\mathbf{x} \in \overline{R}_n$, this estimator yields the unique global maximizer of the likelihood based
on the invariance property of the MLE and standard multinomial theory.

For other values in the sample space, using the results in \citet{wu1983}, the EM sequence of
estimates will be guaranteed to converge to the global maximizer, provided the sequence of estimates
lies in the interior of the parameter space. In our numerical work, there was not a single case
where the final estimate lay on the boundary of the space. To see why this is true, inspection of
$\overline{R}_n^c$ shows that $\mathbf{x}$ is in this set only if $x_{10}$ and $x_{01}$ are both
non-zero. This is sufficient to guarantee the likelihood is maximized at a point with $p_{10}$ and
$p_{01}$ both positive (otherwise the value of $\ell$ will be $-\infty$). If $x_{00} > 0$, the same
theoretical guarantee can be made for $p_{00}$, so that the maximum must lie in the interior and the
EM algorithm will always converge to this point. If $x_{00} =0$, it is more difficult to show
theoretically that $p_{00}> 0$ at the maximum, but our numerical work shows that, even in this case,
the maximum will tend to occur at very large values of $p_{00}$ (for example, with $n=250$, $k=10$,
and $\mathbf{x} = (100, 100, 50)$, we have $\hat{p}_{00}^{MLE} = 0.82$).

To demonstrate the global convergence property of this estimator, Table \ref{tab:sv} gives estimates
and log-likelihood values for the EM algorithm approach compared with numerical optimization on the
full likelihood using the Nelder Mead algorithm \citep{nelder1965}. This was done for the fixed
values $k= 10$, $n = 35$, $\mathbf{x} = (25, 5, 2) \notin \overline{R}_n$, and ten starting values
randomly generated on the probability simplex. We see that, for many of the starting values, both
algorithms yield identical values, but that the EM algorithm-based estimator is extremely consistent
across all initial points. For the Nelder Mead algorithm, however, there is
variation, with some final estimates far from the true maximizer. While in some cases it may be
possible to make a more informed decision about the starting value, the ability to bypass this step
all together, while still guaranteeing convergence to the global maximum, is a strong advantage of
the EM algorithm-based estimator presented here.

\begin{table}[ht]
	\caption{Comparison of estimates and log-likelihood values for the EM algorithm-based estimator with Nelder Mead optimization for ten randomly generated starting values.}
	\label{tab:sv}
	\centering
	\begin{tabular}{ccccccccc}
		\toprule
		&\multicolumn{4}{c}{EM Algorithm} & \multicolumn{4}{c}{Nelder Mead}\\
		\cmidrule(lr){2-5} \cmidrule(lr){6-9}
		Starting Value & $p_{10}$ & $p_{01}$ & $p_{11}$ & $\ell$ & $p_{10}$ & $p_{01}$ & $p_{11}$ & $\ell$\\ 
		\hline
(0.176  0.270  0.429) & 0.139 & 0.022 & 0.000 & -8.737 & 0.140 & 0.023 & 0.000 & -8.738 \\ 
(0.332  0.349 0.244) & 0.139 & 0.022 & 0.000 & -8.737 & 0.139 & 0.022 & 0.000 & -8.737 \\ 
(0.058  0.192  0.164) & 0.139 & 0.022 & 0.000 & -8.737 & 0.139 & 0.022 & 0.000 & -8.737 \\ 
(0.164  0.329  0.213) & 0.139 & 0.022 & 0.000 & -8.737 & 0.140 & 0.022 & 0.000 & -8.737 \\ 
(0.346  0.133  0.271) & 0.139 & 0.022 & 0.000 & -8.737 & 0.139 & 0.022 & 0.000 & -8.737 \\ 
(0.110  0.339  0.065) & 0.139 & 0.022 & 0.000 & -8.737 & 0.139 & 0.022 & 0.000 & -8.737 \\ 
(0.368 0.013  0.364) & 0.139 & 0.022 & 0.000 & -8.737 & 0.162 & 0.027 & 0.000 & -9.205 \\ 
(0.149  0.210  0.262) & 0.139 & 0.022 & 0.000 & -8.737 & 0.238 & 0.023 & 0.000 & -13.202 \\ 
(0.086  0.380  0.307) & 0.139 & 0.022 & 0.000 & -8.737 & 0.139 & 0.022 & 0.000 & -8.737 \\ 
(0.053  0.355  0.202) & 0.139 & 0.022 & 0.000 & -8.737 & 0.139 & 0.022 & 0.000 & -8.737 \\ 
		\bottomrule
	\end{tabular}
\end{table}

\section{Alternative estimators}
In this section we propose two closed form estimators which are alternatives to the MLE given in the previous section which requires numerical optimization.

\subsection{Restricted method of moments estimator}
The first estimator, which is a method of moments type estimator, is motivated by the result in Theorem
\ref{md:mle11}, and simply truncates the value of $p_{11}$ to 0 for values $\mathbf{x} \notin \overline{R}_n$.
This estimator has several advantages, most notably that it has a simple closed form and, as we will show
empirically, slightly outperforms the MLE in terms of both bias and MSE in most cases.

{\small
	\begin{definition}[Restricted method of moments estimator]
		Let $\hat{\mathbf{p}}^{RMM} = (\hat{p}_{10}^{RMM}, \hat{p}_{01}^{RMM}, \hat{p}_{11}^{RMM})$, where
		\begin{align*}
		\hat{p}_{11}^{RMM} &= \max\left\{0, 1 - \left(\frac{x_{00} + x_{10}}{n}\right)^{1 /k} - \left(\frac{x_{00} + x_{01}}{n}\right)^{1 /k} + \left(\frac{x_{00}}{n}\right)^{1 /k}\right\},\\
		\hat{p}_{10}^{RMM} &= 1 - \left(\frac{x_{00} + x_{01}}{n}\right)^{1 /k} - \hat{p}_{11}^{RMM},\\
		\hat{p}_{01}^{RMM} &= 1 - \left(\frac{x_{00} + x_{10}}{n}\right)^{1 /k} - \hat{p}_{11}^{RMM},\\
		\hat{p}_{00}^{RMM} &= 1 - \hat{p}_{10}^{RMM} - \hat{p}_{01}^{RMM} - \hat{p}_{11}^{RMM}.
		\end{align*}
\end{definition}}

It is not hard to see that $\hat{\mathbf{p}}^{RMM} \in \overline{\boldsymbol{\Psi}}_p$ for all $\mathbf{x}$ and
that $\hat{\mathbf{p}}^{RMM} = \hat{\mathbf{p}}^{MLE} \in  \overline{\boldsymbol{\Psi}}_p$ for all $\mathbf{x} \in
\overline{R}_n$. Further properties showing the relation between the RMM estimator and the MLE are given in Section
\ref{sec:tc}.

\subsection{Burrows type estimator}
One of the main advantages of a closed form estimator as above is the ability to provide simple bias
corrections. The problem of bias for the MLE can be generalized from the single-trait group testing
case, where this issue is a major focus of the literature. A proof that no unbiased estimator exists
for the single-trait group testing problem under a fixed sampling model is given in
\cite{haber2018}, and this can be extended directly to the two-trait model considered in this paper.
The issue of bias for the two-trait case is discussed further in \cite{haberM2018}, where it is
shown that any unbiased estimator found under an alternative sampling plan necessarily yields values
outside the parameter space.

In the one-trait case, Burrows' estimator \cite{burrows1987}  has been shown repeatedly to improve on the MLE
in terms of both bias and MSE \citep[see, for example,][]{hepworth2009, ding2016}. The
motivation is to find a shrinkage estimator of the form $(1 -\alpha \bar{x})^{1/k}$ where $\alpha$
is optimal in the sense of removing bias of $O(1 / n)$ from the estimator. \citet{burrows1987}
showed that this is accomplished by taking $\alpha = \frac{n}{n + \eta}$ where
$\eta = \frac{k-1}{2k}$.

In the two disease case, it can be shown that, when the MLE exists in $\boldsymbol{\Psi}_p$,
applying the identical shrinkage coefficient to each term in the estimator yields the same overall
bias reduction. This is true since, for $\mathbf{x} \in \overline{R}_n$, each term of the MLE as
given in Theorem \ref{md:mle1} is marginally binomially distributed, so that the problem is
identical to that in Burrows' original work. In cases where the MLE does not exist in
$\boldsymbol{\Psi}_p$, we can apply the same correction to the terms of the estimator
$\hat{\mathbf{p}}^{RMM}$ to get a similar approximate result on the non-truncated terms.

\begin{definition}[Burrows type estimator]
	Let $\hat{\mathbf{p}}^{B} = (\hat{p}_{10}^{B}, \hat{p}_{01}^{B}, \hat{p}_{11}^{B})$, where
	\begin{align*}
          \hat{p}_{11}^{B} &= \left\{\begin{array}{lc} 1 - \left(\frac{x_{00} + x_{10} + \eta}{n + \eta}\right)^{1 /k} - \left(\frac{x_{00} + x_{01} + \eta}{n + \eta}\right)^{1 /k} + \left(\frac{x_{00} + \eta}{n + \eta}\right)^{1 /k}, & \mathbf{x} \in \overline{R}_n \\
                                       0, & \text{otherwise},\end{array}\right.\\
          \hat{p}_{10}^{B} &= 1 - \left(\frac{x_{00} + x_{01} + \eta}{n + \eta}\right)^{1 /k} - \hat{p}_{11}^{B},\\
          \hat{p}_{01}^{B} &= 1 - \left(\frac{x_{00} + x_{10} + \eta}{n + \eta}\right)^{1 /k} - \hat{p}_{11}^{B},\\
          \hat{p}_{00}^{B} &= 1 - \hat{p}_{10}^{B} - \hat{p}_{01}^{B} - \hat{p}_{11}^{B},
	\end{align*} with $\eta = \frac{k - 1}{2k}$.
\end{definition}

\subsection{Theoretical comparisons of estimators}
\label{sec:tc}
In this section we provide some of the theoretical properties of the three estimators introduced
above in terms of their large sample properties.

\begin{theorem}~\vspace{-15pt}
  \label{thm:distribution}
	\begin{itemize}
		\item[(a)] $\hat{\mathbf{p}}^{RMM} -  \hat{\mathbf{p}}^{MLE} \overset{a.s.}{\to} 0;$
		\item[(b)]
                  $\sqrt{n}(\hat{\mathbf{p}}^{j} - \mathbf{p}) \overset{d}{\to} N(\mathbf{0}, \frac{1}{k^2}\boldsymbol{\Sigma})$ for $j \in \{MLE, RMM, B\}$ where the elements
                  of $\boldsymbol{\Sigma}$ are given in Appendix \ref{sec:apSigma}.
                  	\end{itemize}
      \end{theorem}

			Importantly, this theorem shows that each of the three estimators shares the same asymptotic distribution. 
      The following result gives the first order expectations for each of the three estimators.

      \begin{theorem}~\vspace{-15pt}
        \label{thm:bias}
        {\small
        \begin{itemize}
          \item[(a)]
            \begin{align*}
              \mathrm{E}(\hat{p}_{10}^{MLE}) &= \mathrm{E}(\hat{p}_{10}^{RMM}) = p_{10} + \frac{k - 1}{2k^2n}\left(p_{10} + \frac{1}{p_{00}^{k-1}} - \frac{1}{(p_{00} + p_{10})^{k-1}}\right) + O(n^{-2})\\
              \mathrm{E}(\hat{p}_{01}^{MLE}) &= \mathrm{E}(\hat{p}_{01}^{RMM}) = p_{01} + \frac{k - 1}{2k^2n}\left(p_{01} + \frac{1}{p_{00}^{k-1}} - \frac{1}{(p_{00} + p_{01})^{k-1}}\right) + O(n^{-2})\\
\mathrm{E}(\hat{p}_{11}^{MLE}) &= \mathrm{E}(\hat{p}_{11}^{RMM}) = p_{11} + \frac{k - 1}{2k^2n}\left(p_{11} + \frac{1}{(p_{00} + p_{01})^{k-1}} +  \frac{1}{(p_{00} + p_{10})^{k-1}} - \frac{1}{p_{00}^{k-1}}  - 1\right)\\ &\qquad \qquad \qquad \qquad + O(n^{-2})\\
              \end{align*}
            \item[(b)]
              \begin{align*}
                \mathrm{E}(\hat{p}_{10}^{B}) &= p_{10} + O(n^{-2}) \\
                \mathrm{E}(\hat{p}_{01}^{B}) &= p_{01} + O(n^{-2}) \\
                \mathrm{E}(\hat{p}_{11}^{B}) &= p_{11} + O(n^{-2}) \\
                \end{align*}
            \end{itemize}}
        \end{theorem}

\section{Numerical comparisons}
In this section we provide numerical comparisons for each of the three estimators introduced here in
terms of relative bias and MSE. For the $i^{th}$ component of $\mathbf{p}$ and an estimator
$\hat{\mathbf{p}}$, the relative bias is defined to be
$100 \times \frac{\mathrm{E}(\hat{p}_i - p_i)}{p_i}$. Results are provided for four values of
$\mathbf{p}$, covering a range of realistic scenarios from very small
($\mathbf{p} = (0.001, 0.001. 0.0001)$) to moderately small $(\mathbf{p} = (0.25, 0.05, 0.15))$.
Larger values of the prevalence parameters are not considered here as it would be uncommon for group
testing to be considered in such cases.

Figures \ref{fig:nc1} and \ref{fig:nc2} give the bias and MSE calculations for $k = 2$ and $k = 10$, respectively, for
a fixed number of tests, $n = 25$. More complete numerical comparisons are provided as tables in Appendix \ref{sec:tables} for $n=10, 25, 50,$ and $100$, as well as two additional prevelence points.

\begin{figure}[!htb]
\includegraphics[width = \linewidth]{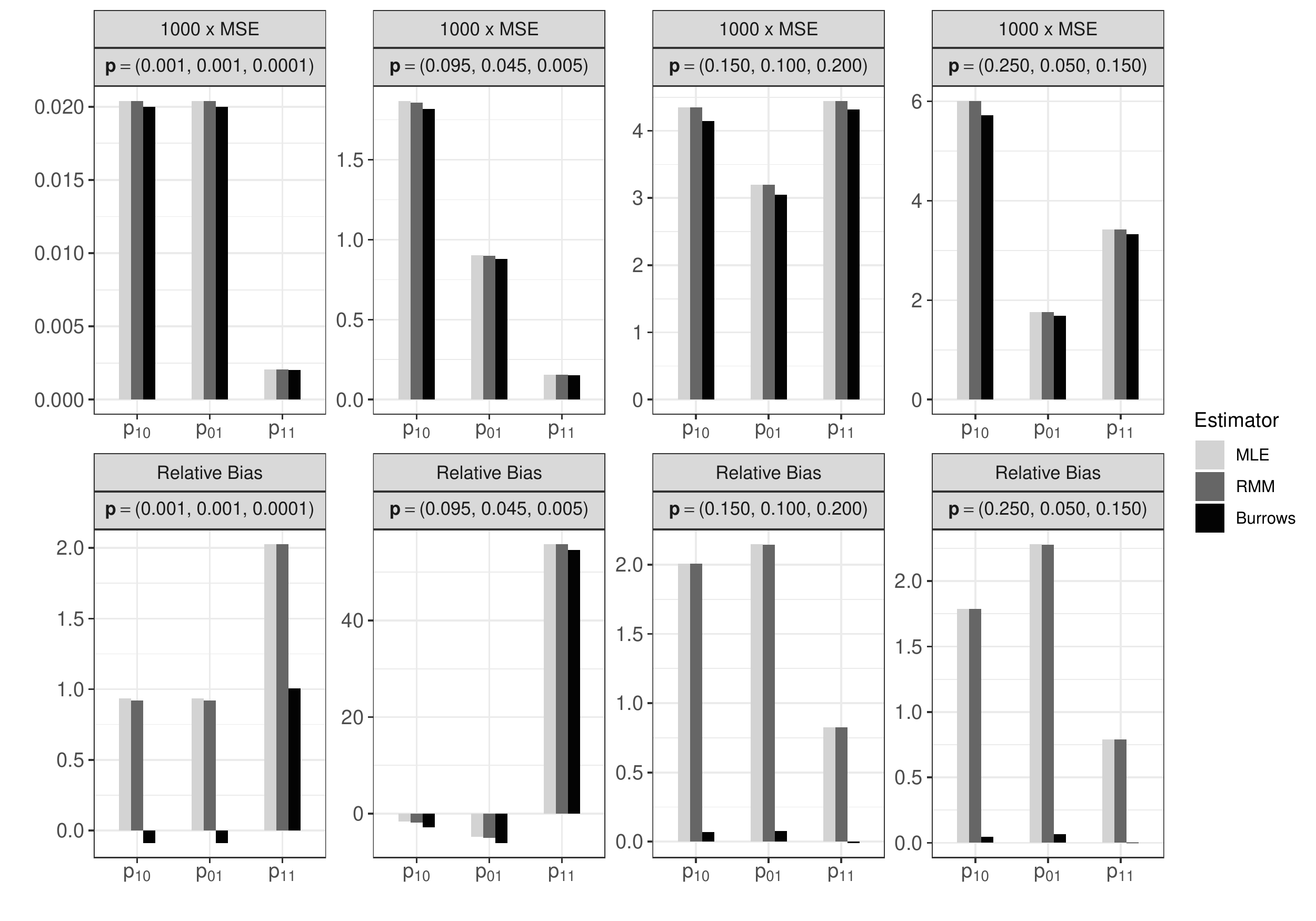}
\caption{Comparisons of 1000 $\times$ mean square error (MSE) and relative bias, defined for the $i^{th}$ element to be $100 \times \frac{\mathrm{E}(\hat{p}_i - p_i)}{p_i}$, for $n = 25$ and $k = 2$. The varying scales in each individual plot should be noted.}
	\label{fig:nc1}
\end{figure}

\begin{figure}[!htb]
\includegraphics[width = \linewidth]{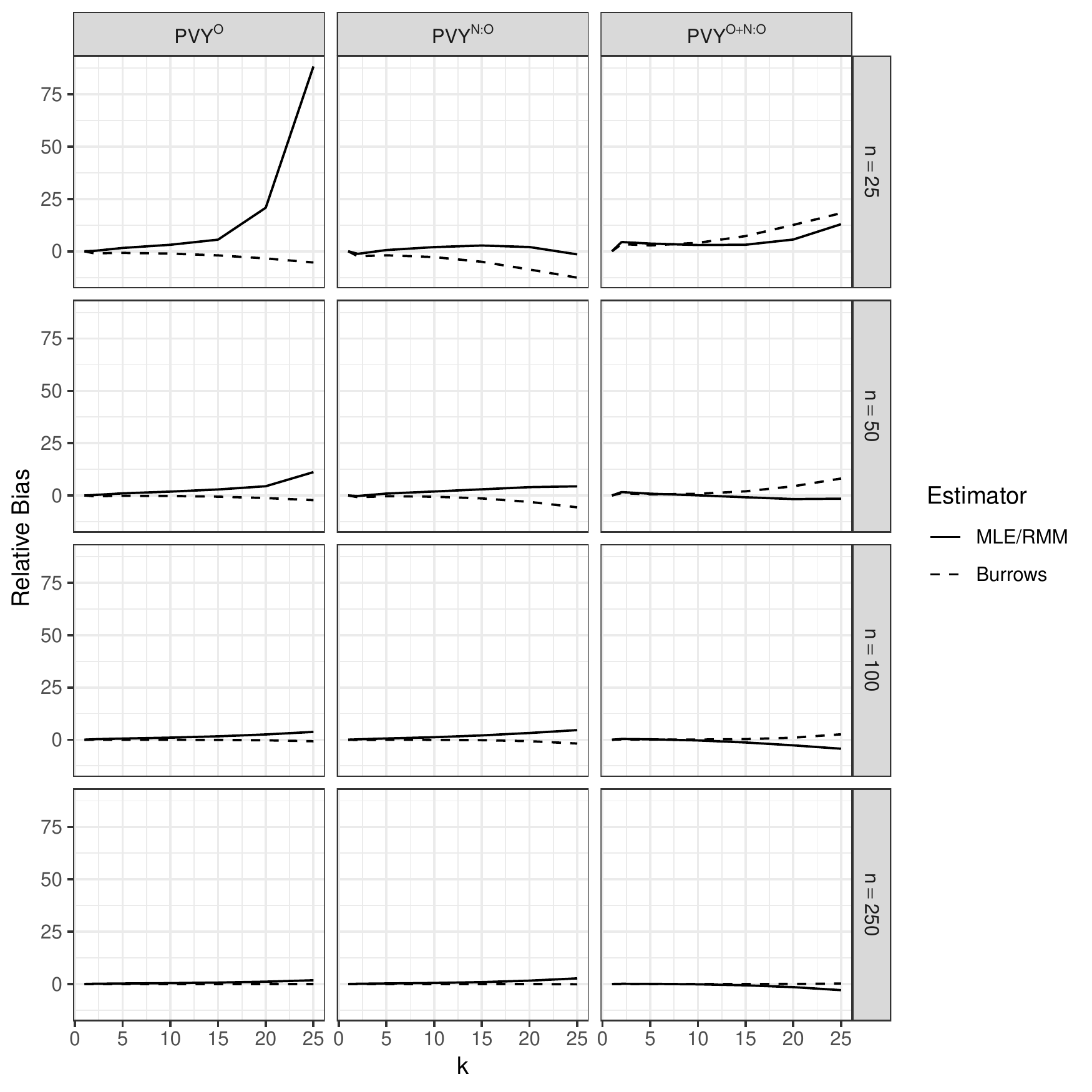}
	\caption{Relative bias, defined for the $i^{th}$ element to be $100 \times \frac{\mathrm{E}(\hat{p}_i - p_i)}{p_i}$, for varying group sizes, $k$, and number of pools tested, $n$, with true $(p_{O}, p_{N:O}, p_{O + N:O})= (0.067,
          0.028, 0.019)$. Note that the values for the MLE and RMM estimators are indistinguishable in this figure.}
	\label{fig:bias1}
\end{figure}

\begin{figure}[!htb]
	\includegraphics[width = \linewidth]{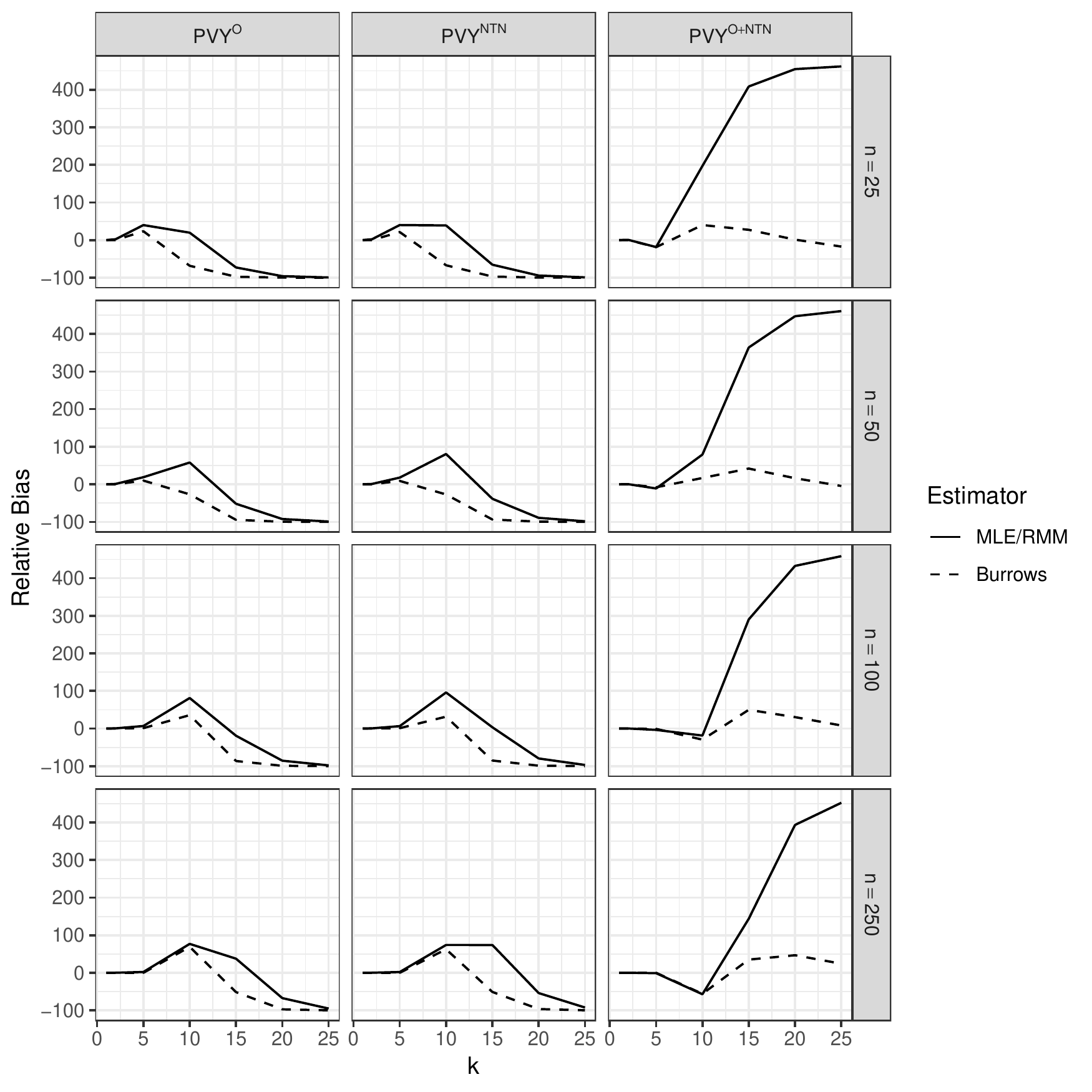}
	\caption{Relative bias, defined for the $i^{th}$ element to be $100 \times \frac{\mathrm{E}(\hat{p}_i - p_i)}{p_i}$, for varying group sizes, $k$, and number of pools tested, $n$, with true $(p_{O}, p_{NTN}, p_{O + NTN})= (0.144,
          0.158, 0.178)$. Note that the values for the MLE and RMM estimators are indistinguishable in this figure.}
	\label{fig:bias2}
      \end{figure}

      \begin{figure}[!htb]
	\includegraphics[width = \linewidth]{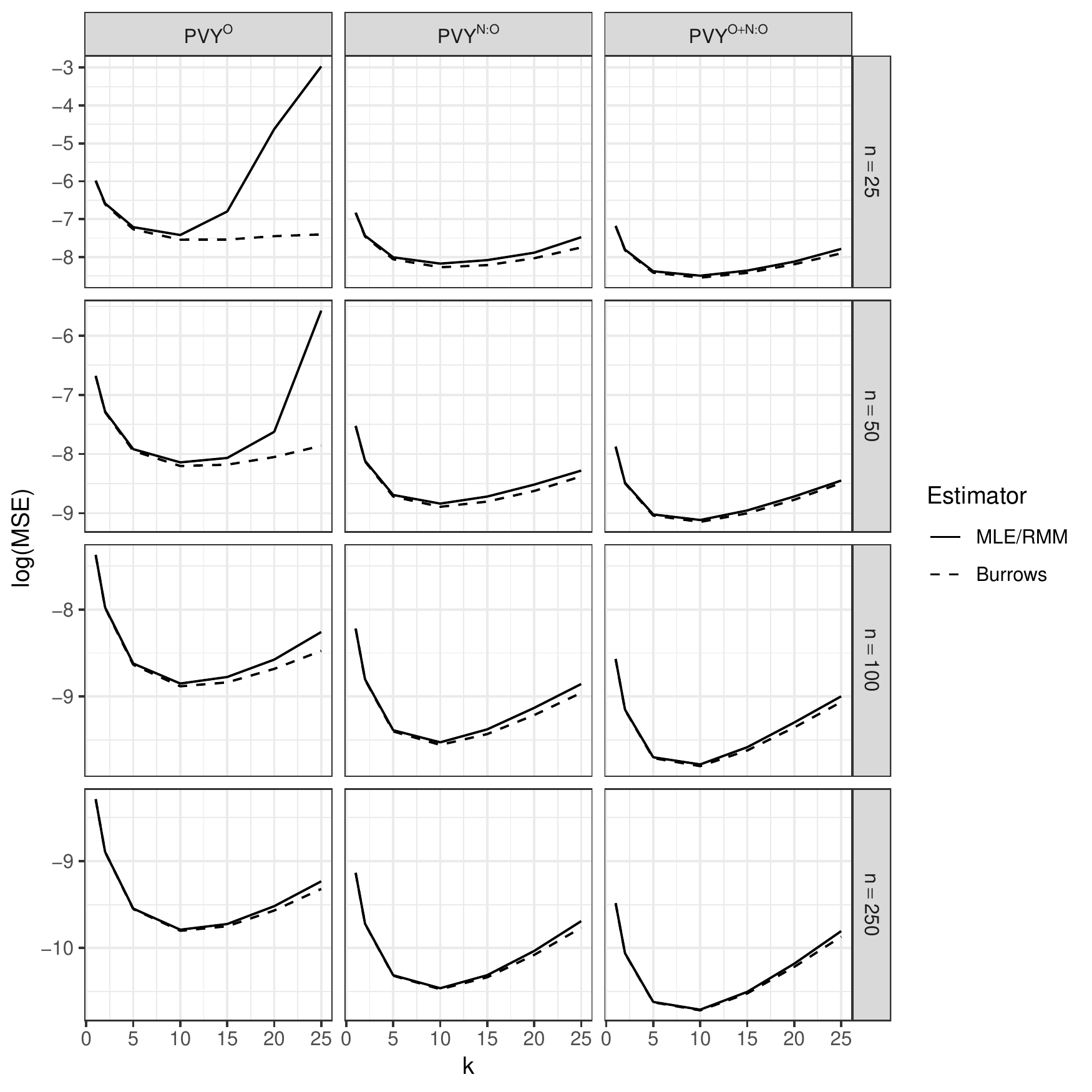}
	\caption{log(MSE) for varying group sizes, $k$, and number of pools tested, $n$, with true $(p_{O}, p_{N:O}, p_{O + N:O})=(0.067, 0.028, 0.019)$. Note that the values for the MLE and RMM estimators are indistinguishable in this figure. The varying scales for the subplots should be noted as well.}
	\label{fig:mse1}
\end{figure}
\begin{figure}[!htb]
	\includegraphics[width = \linewidth]{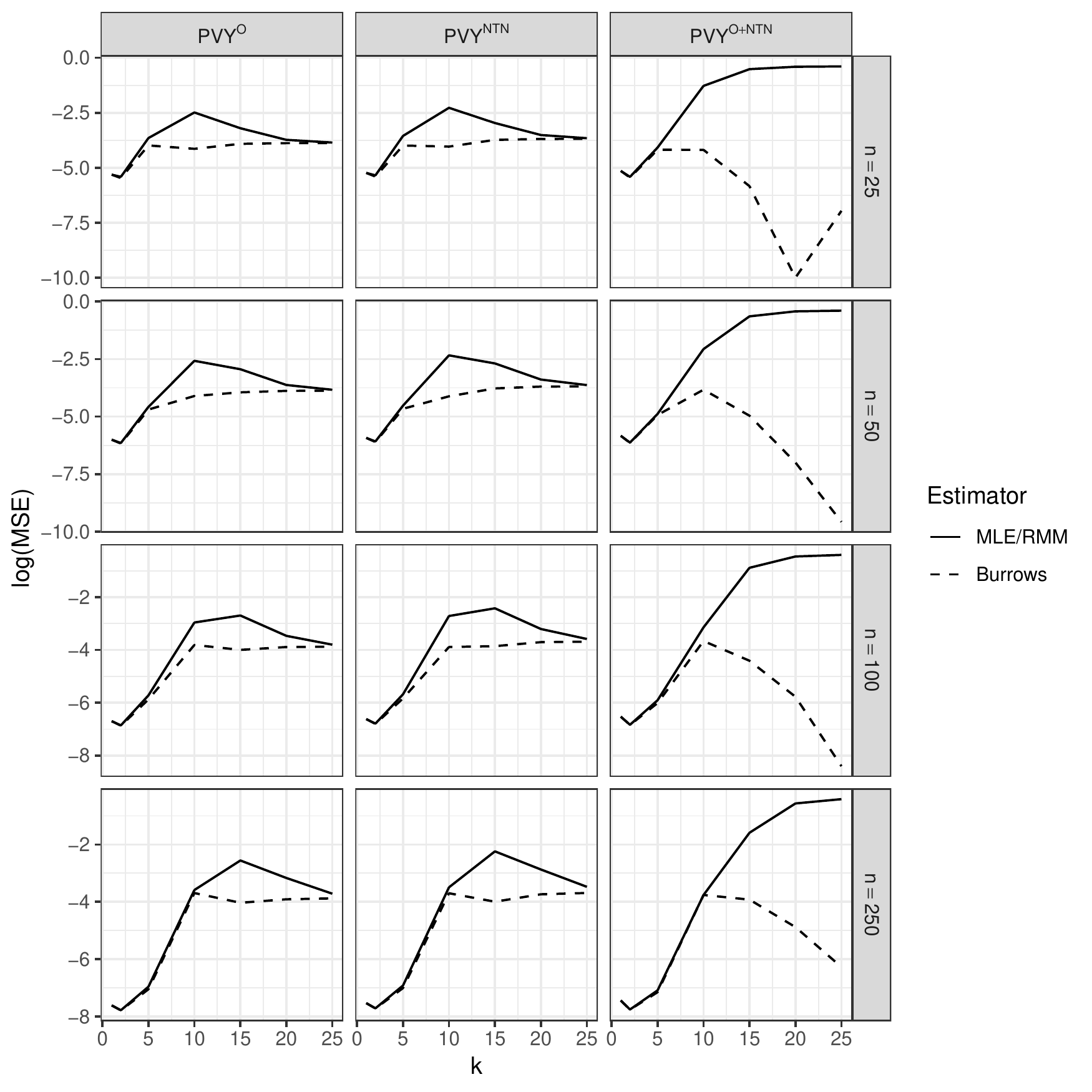}
	\caption{log(MSE) for varying group sizes, $k$, and number of pools tested, $n$, with true $(p_{O},
          p_{NTN}, p_{O + NTN}) = (0.144, 0.158, 0.178)$. Note that the values for the MLE and RMM estimators are
          indistinguishable in this figure. The varying scales for the subplots should be noted as well.}
	\label{fig:mse2}
\end{figure}

From Figure \ref{fig:bias1}, we see that, with the exception of $n=25$, the relative bias is reasonably close to zero for all values of $k$. For the largest number of tests, $n=250$, the bias has all but disappeared. Even for $n=25$ with small $k$, less than or equal to 10, the level of bias is sufficiently small for all estimators. We note that, in this and all subsequent plots, the MLE and RMM values are indistinguishable, which is unsurprising given the results in the previous section.

In Figure \ref{fig:bias2}, where the underlying prevalence is larger,
we see much more variation in the relative bias among estimators and
across $k$. Now, only for small values of $k=1$ or $k=2$ is the bias
near zero. In this case, due to the large value of $p_{O + NTN}$, as
$k$ increases the probability of each test being positive for both
strains of PVY goes to one. As such, we see, for all estimators, the
relative bias for PVY\textsuperscript{O} and PVY\textsuperscript{NTN}
go to $-100$ (since the estimates for each converge to zero). For PVY\textsuperscript{O + NTN}, the MLE and RMM estimates approach $100 \times \left(\frac{1}{0.178} - 1\right) = 461.8$. This is another illustration of the importance of choosing an appropriate pool size, so as to avoid the problem of getting all positive groups \citep[for an excellent discussion of the problem of drawing all positive groups in the one-trait estimation case, see][]{hepworth2009}. While the Burrows type estimator ameliorates this somewhat by shrinking the estimate of $p_{O + NTN}$ towards zero, it does nothing to help the zero estimates for the other parameters.

For the MSE, from Figure \ref{fig:mse1}, we see the $\log(MSE)$ has a similar parabolic shape for each component of
the parameter vector and all $n$. In this particular case, it appears that the same group size, $k = 10$, minimizes
the MSE simultaneously for each parameter and number of tests. By carefully noting the change in scale on the
$y$-axis, it is clear that the MSE is decreasing appreciably with $n$. The values across estimators are similar for
most scenarios, with the Burrows type estimator appearing to offer a slight improvement in MSE. For the
PVY\textsuperscript{O} component, which is the largest among the three, we see that Burrows' estimator is much more
robust to the choice of group size (as seen by the more horizontal nature of its line in the figures), although this advantage decreases with $n$.

For the MSE in the second experiment, Figure \ref{fig:mse2} indicates that, for the MLE and RMM estimators, the minimum value is attained at $k=2$ for all values of $n$ and each parameter. For the Burrows type estimator, this is true for the first two parameters, but not for $p_{O + NTN}$, for which the MSE appears to continue decreasing as a function of $k$ (with the exception of $n=25$, where the minimum occurs at $k=20$). Of course, when considering all three parameters together, the choice of $k=2$ appears ideal, even for the Burrows type estimator. As in Figure \ref{fig:mse2}, the Burrows type estimator outperforms the others, slightly for small $k$ and to a significant degree for $k$ between 10 and 20.

While it is important to look at performance for each component of the parameter vector, the lack of a single index value makes direct quantitative comparison difficult. To address this, we consider looking at the average absolute relative bias and average MSE, where the mean is taken across the three elements of the parameter vector. These values are provided in Tables \ref{tab:avgbias} and \ref{tab:avgmse} for each of the scenarios considered above for the first and second experiments, respectively.

\begin{table}[!htb]
	\centering
	\caption{Average absolute relative bias and average MSE comparisons for PVY example with $(p_{O}, p_{N:O}, p_{O + N:O})= (0.067, 0.028, 0.019)$.}
	\label{tab:avgbias}
	\resizebox{\columnwidth}{!}{%
		\begin{tabular}{llrrrrrrr}
			\toprule
			\multicolumn{3}{l}{Average Absolute Relative Bias}\\
			& & $k=1$ & $k=2$ & $k =5$ & $k=10$ & $k=15$ & $k=20$ & $k=25$ \\ 
			\cmidrule(lr){3-9}
			$n = 25$    & $\hat{\mathbf{p}}^{MLE}$ & 0.000 & 1.938 & 1.942 & 2.691 & 3.804 & 9.472 & 34.247 \\
			&  $\hat{\mathbf{p}}^{RMM}$ & 0.000 & 1.941 & 1.914 & 2.654 & 3.740 & 9.370 & 34.213 \\
			&  $\hat{\mathbf{p}}^{B}$   & 0.000 & 2.265 & 1.833 & 2.636 & 4.751 & 8.244 & 12.063 \\
			$n = 50$   & $\hat{\mathbf{p}}^{MLE}$ & 0.000 & 0.702 & 0.940 & 1.286 & 2.227 & 3.369 & 5.692 \\
			&  $\hat{\mathbf{p}}^{RMM}$ & 0.000 & 0.701 & 0.936 & 1.278 & 2.210 & 3.337 & 5.634 \\
			&  $\hat{\mathbf{p}}^{B}$   & 0.000 & 0.773 & 0.326 & 0.544 & 1.317 & 2.896 & 5.323 \\
			$n = 100$  & $\hat{\mathbf{p}}^{MLE}$ & 0.000 & 0.262 & 0.463 & 0.851 & 1.671 & 2.813 & 4.210 \\
			&  $\hat{\mathbf{p}}^{RMM}$ & 0.000 & 0.261 & 0.463 & 0.851 & 1.668 & 2.806 & 4.194 \\
			&  $\hat{\mathbf{p}}^{B}$   & 0.000 & 0.091 & 0.020 & 0.044 & 0.189 & 0.649 & 1.703 \\
			$n = 250$ & $\hat{\mathbf{p}}^{MLE}$ & 0.000 & 0.104 & 0.184 & 0.353 & 0.736 & 1.401 & 2.470 \\
			&  $\hat{\mathbf{p}}^{RMM}$ & 0.000 & 0.104 & 0.184 & 0.353 & 0.736 & 1.401 & 2.469 \\
			&  $\hat{\mathbf{p}}^{B}$   & 0.000 & 0.000 & 0.000 & 0.001 & 0.001 & 0.017 & 0.136 \\  \midrule
			\multicolumn{3}{l}{Average $1000\times$MSE}\\
			& & $k=1$ & $k=2$ & $k =5$ & $k=10$ & $k=15$ & $k=20$ & $k=25$ \\ 
			\cmidrule(lr){3-9}
			$n = 25$  & $\hat{\mathbf{p}}^{MLE}$ & 1.451 & 0.792 & 0.433 & 0.362 & 0.553 & 3.491 & 17.432 \\
			&  $\hat{\mathbf{p}}^{RMM}$ & 1.451 & 0.790 & 0.433 & 0.361 & 0.550 & 3.487 & 17.427 \\
			&  $\hat{\mathbf{p}}^{B}$   & 1.451 & 0.773 & 0.411 & 0.327 & 0.341 & 0.394 & 0.470 \\
			$n = 50$  & $\hat{\mathbf{p}}^{MLE}$ & 0.725 & 0.398 & 0.218 & 0.182 & 0.202 & 0.284 & 1.418 \\
			&  $\hat{\mathbf{p}}^{RMM}$ & 0.725 & 0.398 & 0.217 & 0.182 & 0.202 & 0.283 & 1.416 \\
			&  $\hat{\mathbf{p}}^{B}$   & 0.725 & 0.393 & 0.212 & 0.173 & 0.184 & 0.218 & 0.274 \\
			$n = 100$ & $\hat{\mathbf{p}}^{MLE}$ & 0.363 & 0.200 & 0.109 & 0.091 & 0.103 & 0.130 & 0.175 \\
			&  $\hat{\mathbf{p}}^{RMM}$ & 0.363 & 0.200 & 0.109 & 0.091 & 0.103 & 0.129 & 0.175 \\
			&  $\hat{\mathbf{p}}^{B}$   & 0.363 & 0.199 & 0.107 & 0.088 & 0.097 & 0.119 & 0.151 \\
			$n =250$ &  $\hat{\mathbf{p}}^{MLE}$ & 0.145 & 0.080 & 0.043 & 0.036 & 0.040 & 0.052 & 0.072 \\
			&  $\hat{\mathbf{p}}^{RMM}$ & 0.145 & 0.080 & 0.043 & 0.036 & 0.040 & 0.052 & 0.072 \\
			&  $\hat{\mathbf{p}}^{B}$   & 0.145 & 0.080 & 0.043 & 0.035 & 0.039 & 0.049 & 0.066 \\
			\bottomrule
	\end{tabular}}
\end{table}

\begin{table}[!htb]
	\centering
	\caption{Average absolute relative bias and average MSE comparisons for PVY example with $(p_{O}, p_{NTN}, p_{O
			+ NTN})= (0.144, 0.158, 0.178)$.}
	\label{tab:avgmse}
	\resizebox{\columnwidth}{!}{%
		\begin{tabular}{llrrrrrrr}
			\toprule
			\multicolumn{3}{l}{Average Absolute Relative Bias}\\
			& & $k=1$ & $k=2$ & $k =5$ & $k=10$ & $k=15$ & $k=20$ & $k=25$ \\ 
			\cmidrule(lr){3-9}
			$n = 25$    & $\hat{\mathbf{p}}^{MLE}$ & 0.000 & 1.645 & 32.815 & 85.262 & 182.278 & 215.017 & 220.157 \\
			&  $\hat{\mathbf{p}}^{RMM}$ & 0.000 & 1.642 & 32.613 & 85.218 & 182.279 & 215.017 & 220.157 \\
			&  $\hat{\mathbf{p}}^{B}$   & 0.000 & 0.057 & 21.208 & 58.466 & 73.960 & 66.875 & 72.426 \\
			$n = 50$   & $\hat{\mathbf{p}}^{MLE}$ & 0.000 & 0.784 & 15.950 & 72.532 & 151.365 & 209.528 & 219.343 \\
			&  $\hat{\mathbf{p}}^{RMM}$ & 0.000 & 0.784 & 15.865 & 72.548 & 151.368 & 209.528 & 219.343 \\
			&  $\hat{\mathbf{p}}^{B}$   & 0.000 & 0.014 & 8.960 & 23.514 & 76.680 & 71.713 & 68.104 \\
			$n = 100$  & $\hat{\mathbf{p}}^{MLE}$ & 0.000 & 0.384 & 5.607 & 65.263 & 104.234 & 199.209 & 217.765 \\
			&  $\hat{\mathbf{p}}^{RMM}$ & 0.000 & 0.384 & 5.594 & 65.363 & 104.234 & 199.209 & 217.765 \\
			&  $\hat{\mathbf{p}}^{B}$   & 0.000 & 0.003 & 1.018 & 32.305 & 73.750 & 76.057 & 69.307 \\
			$n = 250$ & $\hat{\mathbf{p}}^{MLE}$ & 0.000 & 0.152 & 1.738 & 69.306 & 85.050 & 171.500 & 213.232 \\
			&  $\hat{\mathbf{p}}^{RMM}$ & 0.000 & 0.152 & 1.738 & 69.409 & 85.062 & 171.502 & 213.232 \\
			&  $\hat{\mathbf{p}}^{B}$   & 0.000 & 0.001 & 0.033 & 62.471 & 45.787 & 80.101 & 74.528 \\ \midrule
			\multicolumn{3}{l}{Average $1000\times$MSE}\\
			& & $k=1$ & $k=2$ & $k =5$ & $k=10$ & $k=15$ & $k=20$ & $k=25$ \\ 
			\cmidrule(lr){3-9}
			$n = 25$  & $\hat{\mathbf{p}}^{MLE}$ & 5.368 & 4.565 & 23.864 & 154.920 & 229.355 & 239.482 & 240.586 \\
			&  $\hat{\mathbf{p}}^{RMM}$ & 5.368 & 4.563 & 23.753 & 154.905 & 229.355 & 239.482 & 240.586 \\
			&  $\hat{\mathbf{p}}^{B}$   & 5.368 & 4.359 & 17.461 & 16.267 & 15.627 & 15.219 & 15.584 \\
			$n = 50$  & $\hat{\mathbf{p}}^{MLE}$ & 2.684 & 2.199 & 9.558 & 99.728 & 215.451 & 237.782 & 240.389 \\
			&  $\hat{\mathbf{p}}^{RMM}$ & 2.684 & 2.199 & 9.508 & 99.740 & 215.451 & 237.782 & 240.389 \\
			&  $\hat{\mathbf{p}}^{B}$   & 2.684 & 2.154 & 8.617 & 18.140 & 16.418 & 15.440 & 15.282 \\
			$n = 100$ & $\hat{\mathbf{p}}^{MLE}$ & 1.342 & 1.082 & 3.134 & 53.677 & 189.030 & 234.070 & 239.949 \\
			&  $\hat{\mathbf{p}}^{RMM}$ & 1.342 & 1.082 & 3.125 & 53.740 & 189.026 & 234.070 & 239.949 \\
			&  $\hat{\mathbf{p}}^{B}$   & 1.342 & 1.071 & 2.726 & 22.721 & 17.202 & 16.035 & 15.318 \\
			$n =250$ &  $\hat{\mathbf{p}}^{MLE}$ & 0.537 & 0.429 & 0.921 & 27.087 & 129.438 & 222.342 & 238.479 \\
			&  $\hat{\mathbf{p}}^{RMM}$ & 0.537 & 0.429 & 0.921 & 27.156 & 129.445 & 222.341 & 238.479 \\
			&  $\hat{\mathbf{p}}^{B}$   & 0.537 & 0.427 & 0.850 & 24.285 & 18.635 & 17.142 & 15.823 \\
			\bottomrule
	\end{tabular}}
\end{table}

Table \ref{tab:avgbias} shows that, for all estimators, the smallest average MSE value occurs when $k=10$. This adds support to what was observed in Figure \ref{fig:mse1}. With this group size, the average MSE is $4.4$ times smaller when $n=25$ than the value when group testing is not used. For $n=250$, this decreases to $4.1$, which is still a very large gain in efficiency. This is achieved with only a small increase in the bias with, for example, the Burrows type estimator yielding relative bias values of $2.636$ for $n=25$ and $0.001$ when $n=250$. This shows a clear advantage to using group testing in similar applications. Even if the ideal group size of $10$ was not chosen, we still see a large benefit for all estimators when $k=2$ or $k=5$. For the Burrows type estimator, even if too large of a group size, such as $k=25$, were chosen, the result is still much better than when group testing is not used. For example, with $n=50$ and $k=25$, the Burrows type estimator yields an average MSE $2.6$ times smaller than that achieved without group testing, and an average absolute relative bias of only $5.3$. Of course, in the same scenario if the MLE were used instead, the average MSE would be nearly double that of the non-group testing case, showing again the benefit of the Burrows type estimator.

In Table \ref{tab:avgmse} we again see, similar to what was indicated in Figure \ref{fig:mse2}, that the best choice of group size for the second experiment is $k=2$. In this case, however, the gains relative to the non grouping case are much more modest with, for example, the average MSE for the Burrows type estimator only $1.2$ times smaller when $n = 25$. 

\section{Discussion}
In this paper, we have addressed the problem of prevalence estimation for two-traits simultaneously using group-testing methods. While a solution for large samples is straightforward (as shown in Theorem \ref{md:mle1}), the high probability of the MLE lying on the boundary for small sample cases requires more careful consideration. We have shown that (see Table \ref{tab:prob1}), depending on the underlying prevalence, this problem can be substantial, even for sample sizes that are much larger than would be feasible in many applications.

The problem of finding an MLE has been shown as a special case of maximizing a multinomial likelihood with a restricted parameter space. While this is, in general, a difficult problem, we showed that, whenever a closed form maximizer does not exist, the optimization problem can be expressed as a simpler problem in one fewer dimension. This was used to develop an estimation approach based on the EM algorithm which is simple both conceptually and computationally. More importantly, the resulting estimator is guaranteed to converge to the global maximizer for all values in the sample space.

In addition, we have provided a second estimator, based on the method of moments, which very closely approximates the MLE, but has the advantage of a closed form expression. This, in turn, was used to develop a third estimator based on the shrinkage estimator for one-trait group testing estimation presented in \citet{burrows1987}.

Among the three estimators, numerical comparisons showed that none uniformly outperforms the others in terms of relative bias and MSE. Still, the Burrows type estimator does generally offer some advantage in terms of MSE, and has the added benefit of being robust to poor specification of the group size. As such, this estimator can be recommended as the best choice in most cases, including those similar to the PVY strain prevalence estimation application considered here.

For many realistic cases, such as with small $\mathbf{p}$ and moderate $n$, our numerical results indicate that the
bias is well contained even though the number of tests is too small to rely on large sample results. It is in these
cases that the importance of the results presented here is highlighted, since the use of the large sample MLE or
direct numerical optimization can yield poor results. For other cases, such as with very small $n$ or larger
parameter values, the bias is much larger and the estimators presented here are not ideal. Unfortunately, there
exist no alternatives in the literature to date for addressing these issues, so that future research in this area
is essential.

One important extension that we have not considered here is estimation when tests are subject to misclassification.
Among many other examples, group testing studies incorporating testing errors can be found in \citet{tu1995},
\citet{liu2012}, and \citet{zhang2014}. While very common in medical studies, this issue appears to be less
commonly considered in plant science areas such as the PVY application presented here. For example, none of the
previously cited papers in the PVY literature \citep[][]{gray2010, mello2011, mondal2017} included
misclassification parameters in their statistical models, although exceptions do exist \citep[see, as one
example,][]{liu2011}. Reasons for this include assumptions that the testing errors are negligible in a given
application, as well as sample sizes that are too small for simultaneously estimating the usually unknown
misclassification parameters. When testing errors are incorporated into the model, the results
presented in this paper which greatly simplify estimation (e.g., Theorem \ref{md:mle11}) do not hold, so that
different approaches must be taken. It is still possible to numerically optimize the likelihood in such cases,
although care must be taken with the starting values when this is done.

A second area which has not been considered here, is how to best choose the group size $k$. In the one-trait estimation case, this has proven to be a very difficult problem with solutions requiring either reasonably precise prior knowledge of the true parameter value \citep[e.g.,][]{swallow1985} or adaptive approaches \citep[e.g.,][]{hos1994}. For two traits, an adaptive approach based on optimal design theory has been given \citep{ho2000}, but relies heavily on large sample assumptions. Numerical studies or theoretical results for constructing locally optimum designs based on prior information have, to our knowledge, not been done. Further research is necessary to provide reasonable small sample solutions to this problem in order to realize the full benefits of group-testing for a wide range of applications.

\appendix

\section{Large sample covariance matrix for all estimators}
\label{sec:apSigma}
Letting $\lambda_{10} = p_{00} + p_{10}$ and $\lambda_{01} = p_{00} + p_{01}$, the elements of $\boldsymbol{\Sigma}$ as defined in Theorem \ref{thm:distribution} (b) are given by:
                  {\small
                  \begin{align*}
                    \boldsymbol{\Sigma}_{11} &= p_{10}^2\left(\frac{1}{\lambda_{10}^k} - 1\right) + p_{00}^2\left(\frac{1}{p_{00}^k} - \frac{1}{\lambda_{10}^k}\right) \\
                    \boldsymbol{\Sigma}_{22} &= p_{01}^2\left(\frac{1}{\lambda_{01}^k} - 1\right) + p_{00}^2\left(\frac{1}{p_{00}^k} - \frac{1}{\lambda_{01}^k}\right) \\
                    \boldsymbol{\Sigma}_{21} &= p_{10}p_{01}\left(\frac{p_{00}^k}{\lambda_{10}^{k}\lambda_{01}^{k}} - 1\right) + p_{00}p_{10}\left(\frac{p_{00}^k}{\lambda_{10}^{k}\lambda_{01}^{k}} - \frac{1}{\lambda_{10}^k} \right)
                                               + p_{00}p_{01}\left(\frac{p_{00}^k}{ \lambda_{10}^{k}\lambda_{01}^{k}} - \frac{1}{\lambda_{01}^k} \right)\\
                    &\qquad + p_{00}^2\left(\frac{p_{00}^k}{\lambda_{10}^{k}\lambda_{01}^{k}} - \frac{1}{\lambda_{10}^{k}} - \frac{1}{\lambda_{01}^{k}} + \frac{1}{p_{00}^{k}}\right)\\
                    \boldsymbol{\Sigma}_{31} &= p_{10}^2\left(1 - \frac{1}{\lambda_{10}^k}\right) + (p_{10}p_{01} + p_{00}p_{10})\left(1 - \frac{p_{00}^k}{\lambda_{10}^k\lambda_{01}^k}\right)
                                               + p_{00}p_{01}\left(\frac{1}{\lambda_{01}^k} - \frac{p_{00}^k}{\lambda_{10}^k\lambda_{01}^k}\right) \\
                                               &\qquad + p_{00}^2\left(\frac{1}{\lambda_{10}^k}  + \frac{1}{\lambda_{01}^k} - \frac{p_{00}^k}{\lambda_{10}^k\lambda_{01}^k} - \frac{1}{p_{00}^k}\right)\\
                    \boldsymbol{\Sigma}_{32} &= p_{01}^2\left(1 - \frac{1}{\lambda_{01}^k}\right) + (p_{10}p_{01} + p_{00}p_{01}) \left(1 - \frac{p_{00}^k}{\lambda_{10}^k\lambda_{01}^k}\right)
                                               + p_{00}p_{10}\left(\frac{1}{\lambda_{10}^k} - \frac{p_{00}^k}{\lambda_{10}^k\lambda_{01}^k}\right) \\
                                               &\qquad + p_{00}^2\left(\frac{1}{\lambda_{10}^k}  + \frac{1}{\lambda_{01}^k} - \frac{p_{00}^k}{\lambda_{10}^k\lambda_{01}^k} - \frac{1}{p_{00}^k}\right) \\
                    \boldsymbol{\Sigma}_{33} &= p_{10}^2\left(\frac{1}{\lambda_{10}^k} - 1\right) + p_{01}^2\left(\frac{1}{\lambda_{01}^k} - 1\right)
                                               + 2(p_{10}p_{01} + p_{00}p_{10} + p_{00}p_{01})\left(\frac{p_{00}^k}{\lambda_{10}^k\lambda_{01}^k} - 1\right) \\
                    &\qquad +  p_{00}^2\left(\frac{2p_{00}^k}{\lambda_{10}^k\lambda_{01}^k} + \frac{1}{p_{00}^k} - \frac{1}{\lambda_{10}^k} - \frac{1}{\lambda_{01}^k} - 1\right).
                    \end{align*}}
\section{Proofs}
\label{sec:ap}

\subsection{Proof of Lemma \ref{lm:llcon}}

(a) We proceed by showing the negative Hessian, $-H(\mathbf{p})$, to be
positive definite. For convenience, we consider the parameter vectors
$\tilde{\btheta} = (\theta_{00}, \theta_{10}, \theta_{01}) = (p_{00}^k, (p_{00} + p_{10})^k - p_{00}^k, (p_{00} +
p_{01})^k - p_{00}^k)$ and $\tilde{\mathbf{p}} = (p_{00}, p_{10}, p_{01})$.
For the standard multinomial vector, we have $-H(\tilde{\btheta}) = \mathbf{D} +
\frac{x_{11}}{\theta_{11}^2} \mathbf{11^\prime}$ where $\mathbf{D} =
diag\left(\frac{x_{00}}{\theta_{00}^2}, \frac{x_{10}}{\theta_{10}^2},
\frac{x_{01}}{\theta_{01}^2}\right)$ which is positive definite since, for
any $\mathbf{z} = (z_1, z_2, z_3) \in \mathbb{R}^3$ such that $\mathbf{z} \neq \mathbf{0}$,
\begin{equation}
-\mathbf{z}^\prime H(\tilde{\btheta}) \mathbf{z} =
\frac{x_{00}z_1^2}{\theta_{00}^2} + \frac{x_{10}z_2^2}{\theta_{10}^2} +
\frac{x_{01}z_3^2}{\theta_{01}^2} + \frac{x_{11}(\mathbf{1^\prime
		z})^2}{\theta_{11}^2} > 0.
\end{equation}

Now, $H(\tilde{\mathbf{p}}) = \frac{\partial \btheta}{\partial
	\mathbf{p}}^\prime H(\tilde{\btheta}) \frac{\partial \btheta}{\partial
	\mathbf{p}}$, where
\[\frac{\partial \btheta}{\partial
	\mathbf{p}} = k\left( \begin{array}{ccc} p_{00}^{k-1} & 0 & 0 \\ (p_{00} +
p_{10})^{k-1} - p_{00}^{k-1} & (p_{00} +
p_{10})^{k-1} & 0
\\
(p_{00} + p_{01})^{k-1} - p_{00}^{k -1} & 0 &
(p_{00}
+
p_{01})^{k-1}
\\
\end{array}\right),\]
and so
\[-\mathbf{z}^\prime H(\tilde{\mathbf{p}}) \mathbf{z} = -\mathbf{z^*}^\prime
H(\tilde{\btheta}) \mathbf{z^*} > 0\] by (4) provided
\begin{equation}
\label{eq:zstar}
\mathbf{z^*} = \frac{\partial \btheta}{\partial
	\mathbf{p}} \mathbf{z} \neq \mathbf{0}.
\end{equation}

Since $\frac{\partial \btheta}{\partial
	\mathbf{p}}$ is full rank, (\ref{eq:zstar}) holds whenever $\mathbf{z}
\neq \mathbf{0}$ and the result follows.

(b) If $\mathbf{x} \in \mathcal{X} \cap \mathcal{X}_0^c$, so that at least
one element is zero, then the inequality in (4) is replaced
by \[-\mathbf{z}^\prime H(\tilde{\btheta}) \mathbf{z}  \geq 0.\] The
remainder of the proof is identical to that of (a) with the result that
$-H(\mathbf{p})$ is positive semi-definite, hence the log-likelihood is
concave (though not necessarily strict).

\subsection{Proof of Theorem \ref{md:mle1}}
  (a) This follows directly from the previous discussion using the invariance property based on the multinomial MLE.
	
	(b) Let $g(\mathbf{x}) = \left(\frac{x_{00} + x_{10}}{n}\right)^{1 /k} + \left(\frac{x_{00} + x_{01}}{n}\right)^{1 /k} - \left(\frac{x_{00}}{n}\right)^{1 /k}$. Then, by the strong law of large numbers \[g(\mathbf{x}) \overset{a.s.}{\to} \left(\theta_{00} + \theta_{10}\right)^{1 /k} + \left(\theta_{00} + \theta_{01}\right)^{1 /k} - \left(\theta_{00}\right)^{1 /k} = p_{00} + p_{10} + p_{01} < 1,\] so that for large enough $n$, $P(\mathbf{x} \in R_n) = 1$. Likewise, $\mathbf{x} / n \overset{a.s.}{\to} \btheta$ implies that $\mathbf{x}$ lies in the interior of its support with probability one. That is, for large enough $n$, $P(\mathbf{x} \in \mathcal{X}_0) = 1$. It follows then that $P(\mathbf{x} \in \mathcal{X}_0 \cap R_n) = 1$.

	\subsection{Proof of Theorem \ref{md:mle11}}
  For all $\mathbf{x}$, the invariance property of the MLE gives the unique maximizer over $\boldsymbol{\Psi}_\theta$ to be $\tilde{\mathbf{p}} = h(\bar{\mathbf{x}})$, where $h$ is as in Lemma \ref{lm:h24}.
	If $\mathbf{x} \in \mathcal{X}_0 \cap \overline{R}_n^c$ then we have $\tilde{p}_{11} < 0$, by the definition of $\overline{R}_n$. Furthermore, by Lemma \ref{lm:llcon} (a), the log-likelihood, $\ell$, is strictly concave over this set and, by (\ref{eq:theta14}) and (\ref{eq:theta24}), each of $\tilde{p}_{00}, \tilde{p}_{10},$ and $\tilde{p}_{01}$ are non-negative (since otherwise, the corresponding $\theta$ values would be negative, hence not in $\boldsymbol{\Psi}_\theta$). Suppose now that $\mathbf{p}^\prime \in \overline{\boldsymbol{\Psi}}_p$ is the true maximizer over the boundary and satisfies $p^\prime_{11} > 0$. Then, since $\tilde{\mathbf{p}}$ is a global maximum and $\ell$ is strictly concave, the line segment $\ell(t\tilde{\mathbf{p}} + (1 - t)\mathbf{p}^\prime),\ 0 \leq t \leq 1$ is strictly decreasing as $t \to 0$. However, there exists a point in $\partial \boldsymbol{\Psi}$ on the line with $t > 0$, say $\mathbf{p}^{\prime\prime}$, satisfying $p^{\prime\prime}_{11} = 0$ and $\ell(\mathbf{p}^{\prime\prime}) > \ell(\mathbf{p}^\prime)$.
	
	If $\mathbf{x} \notin \overline{R}_n$ and at least one of the elements of
	$\mathbf{x}$ is equal to zero, the log-likelihood is concave (not strictly, Lemma \ref{lm:llcon} (b)).
	As such, the previous analysis can be repeated with the conclusion that
	$\ell(\mathbf{p}^{\prime\prime}) \geq \ell(\mathbf{p}^\prime)$, so that
	$\ell$ is maximized at a point with $p_{11} = 0$, although perhaps not
	uniquely.
	
	For $x_{11} = 0$, the log-likelihood is proportional to
	$x_{00}\log(p_{00}^k) + x_{10}\log((p_{00} + p_{10})^k - p_{00}^k) +
	x_{01}\log((p_{00} + p_{01})^k - p_{00}^k)$. Let $\mathbf{p}$ be any point
	such that $p_{11} > 0$. It is clear that taking the point $\mathbf{p}^\prime
	= (p_{00}, p_{10} + p_{11}, p_{01}, 0)$ the value of this function can be
	increased, provided $x_{10} > 0$. If instead $x_{10} = 0$ but $x_{01} > 0$,
	the same can achieved by taking $\mathbf{p}^\prime
	= (p_{00}, p_{10}, p_{01} + p_{11}, 0)$. Likewise, if both $x_{10} = 0$ and
	$x_{01} = 0$, so that $x_{00} = n$, the log-likelihood can be increased by
	taking $\mathbf{p}^\prime = (p_{00} + p_{11}, p_{10}, p_{01}, 0)$. Since one
	of these three cases must occur, it follows that any point maximizing the
	log-likelihood will necessarily have $p_{11} = 0$.

\subsection{Proof of Lemma \ref{lm:llcon2}}

The proof of both (a) and (b) is nearly identical to that of Lemma
\ref{lm:llcon}, the primary difference being that we now have
{\small
\[\frac{\partial \btheta}{\partial
	\mathbf{p}} = k\left( \begin{array}{cc}
-(1 - p_{10} - p_{01})^{k-1}  &-(1 - p_{10} -
p_{01})^{k-1} \\
(1 - p_{10} - p_{01})^{k-1} & (1 - p_{10} -
p_{01})^{k-1} - (
1 - p_{01})^{k-1}
\\
(1 - p_{10} - p_{01})^{k-1} - (1 - p_{10})^{k-1}
& (1 - p_{10} - p_{01})^{k-1}\\
\end{array}\right).\]}
Since this matrix has full column rank, it follows that (\ref{eq:zstar})
holds if and only if $\mathbf{z} \neq \mathbf{0}$ and the rest of the proof is identical.

\subsection{Proof of Result \ref{res:EM}}
	Let $\mathbf{Z}^{ij} = (z_{10}^{ij}, z_{01}^{ij})$ represent the (latent) disease status of the $j^{th}$ unit from the $i^{th}$ pool, so that $\mathbf{Z}^{ij} \sim MN(1, \mathbf{p}^*)$ and set $z_{00}^{ij} = n - z_{10}^{ij} - z_{01}^{ij}$.
	
	Using $\mathbf{Z}$ as the complete data in the EM framework, the complete data log-likelihood is given by
	\[\ell_C(\mathbf{p}^*, \mathbf{z}) \propto \sum \sum z_{00}^{ij} \log (1 - p_{10} - p_{01}) + \sum \sum z_{10}^{ij} \log (p_{10}) + \sum \sum z_{01}^{ij} \log (p_{01}).\] 
	
	We now proceed to calculate the E and M steps, respectively.
	\begin{description}
		\item[E-step:] \mbox{} \newline
		Let $\zeta_{r}^s(\mathbf{p}^{*(t)}) = \mathrm{E}(Z_r^{11} | \vartheta_s^{1} = 1,
                \mathbf{p}^{*(t)}),$ where $r \in \{(00),
		(10), (01)\}$ is the true status (under the reduced model), and $s \in S = \{(00), (10),
                (01), (11)\}$ is the observed status, and $\mathbf{p}^{*(t)}$ is the parameter estimate
		at the $t^{th}$ iteration. Then, we have the expectation of the
		complete log-likelihood
		\begin{align}
		\label{eq:Q}
		Q\left(\mathbf{p}^*; \mathbf{p}^{*(t)}\right) &\propto k \sum_{s \in S} \zeta_{00}^s(\mathbf{p}^{*(t)})X_s \log (1 - p_{10} - p_{01}) + k
		\sum_{s \in S} \zeta_{10}^s(\mathbf{p}^{*(t)})X_s \log (p_{10}) \nonumber \\ &\qquad + k \sum_{s \in S}
		\zeta_{01}^s(\mathbf{p}^{*(t)})X_s \log (p_{01}).
		\end{align}
		
		To calculate the values of $\zeta_{r}^s(\mathbf{p}^{*(t)})$, we have
		\begin{align*}
		\zeta_{00}^{00}(\mathbf{p}^{*(t)})  &= \mathrm{E}(Z_{00}^{11}|\vartheta_{00}^{1} = 1, \mathbf{p}^{*(t)}) = P(Z_{00}^{11} = 1|\vartheta_{00}^{1}=1)\\
		&= \frac{P(\vartheta_{00}^1 = 1| Z_{00}^{11} = 1)P(Z_{00}^{11} = 1)}{P(\vartheta_{00}^{1} = 1)} =\frac{(p_{00}^{(t)})^{k-1}p_{00}^{(t)}}{\theta_{00}^{(t)}} \\
		&= 1,
		\end{align*}
                where the fourth equality holds since, conditioning on the first observation being negative, the
                pool will be negative if and only if the remaining $k-1$ units are negative as well.

		Likewise,
		
		\[\zeta_{00}^{10}(\mathbf{p}^{*(t)}) = \frac{[(p_{00}^{(t)} + p_{10}^{(t)})^{k-1} -
			(p_{00}^{(t)})^{k-1}]p_{00}^{(t)}}{\theta_{10}^{(t)}} \]
		\[\zeta_{00}^{01}(\mathbf{p}^{*(t)}) = \frac{[(p_{00}^{(t)} + p_{01}^{(t)})^{k-1} -
			(p_{00}^{(t)})^{k-1}]p_{00}^{(t)}}{\theta_{01}^{(t)}} \]
		\[\zeta_{00}^{11}(\mathbf{p}^{*(t)}) = \frac{[1 - (p_{00}^{(t)} + p_{10}^{(t)})^{k-1} - (p_{00}^{(t)} + p_{01}^{(t)})^{k-1} 
			+ (p_{00}^{(t)})^{k-1}]p_{00}^{(t)}}{\theta_{11}^{(t)}} \]
		\[\zeta_{10}^{00}(\mathbf{p}^{*(t)}) = 0\]
		\[\zeta_{10}^{10}(\mathbf{p}^{*(t)}) = \frac{(p_{00}^{(t)} + p_{10}^{(t)})^{k-1}p_{10}^{(t)}}{\theta_{10}^{(t)}} \]
		\[\zeta_{10}^{01}(\mathbf{p}^{*(t)}) = 0 \]
		\[\zeta_{10}^{11}(\mathbf{p}^{*(t)}) =  \frac{[1 - (p_{00}^{(t)} + p_{10}^{(t)})^{k-1}]p_{10}^{(t)}}{\theta_{11}^{(t)}} \]
		\[\zeta_{01}^{00}(\mathbf{p}^{*(t)}) = 0\]
		\[\zeta_{01}^{10}(\mathbf{p}^{*(t)}) = 0 \] 
		\[\zeta_{01}^{01}(\mathbf{p}^{*(t)}) =  \frac{(p_{00}^{(t)} + p_{01}^{(t)})^{k-1}p_{01}^{(t)}}{\theta_{01}^{(t)}} \]
		\[\zeta_{01}^{11}(\mathbf{p}^{*(t)}) =  \frac{[1 - (p_{00}^{(t)} + p_{01}^{(t)})^{k-1}]p_{01}^{(t)}}{\theta_{11}^{(t)}}.\]

		\item[M-step:] \mbox{} \newline
		Since (\ref{eq:Q}) is a standard multinomial log-likelihood of size $kn$, the unique global maximizer
		is given by
		\begin{align*}
		p_{00}^{(t + 1)} &= \frac{\sum_{s \in S} \zeta_{00}^s(\mathbf{p}^{*(t)})X_s}{n} \\
		p_{10}^{(t + 1)} &= \frac{\sum_{s \in S} \zeta_{10}^s(\mathbf{p}^{*(t)})X_s}{n} \\
		p_{01}^{(t + 1)} &= \frac{\sum_{s \in S} \zeta_{01}^s(\mathbf{p}^{*(t)})X_s}{n}.
		\end{align*}

              \end{description}

\subsection{Lemmas for Theorems \ref{thm:distribution} and \ref{thm:bias}}
Before proving Theorems \ref{thm:distribution} and \ref{thm:bias}, we provide the following three lemmas. 

\begin{lemma}
  \label{lm:prn}
$P(\mathbf{x} \in \overline{R}_n^c) = O(n^{-2})$
\end{lemma}

\begin{proof}
We have $\mathbf{x} \in \overline{R}_n^c$ only if $\left(\frac{x_{00} + x_{10}}{n}\right)^{1 / k} +
\left(\frac{x_{00} + x_{01}}{n}\right)^{1 / k} - \left(\frac{x_{00}}{n}\right)^{1 / k} > 1$. Choose
$\epsilon$ such that $(\theta_{00} + \theta_{10} + \epsilon)^{1 / k} + (\theta_{00} + \theta_{01} +
\epsilon)^{1/k} - (\theta_{00} - \epsilon)^{1 / k} \leq 1$ so that $\mathbf{x} \in
\overline{R}_n^c$ implies the event $\{n\theta_{00} - x_{00} > n \epsilon\}\cup \{x_{00} + x_{10} - n\theta_{00} - n\theta_{10} >
n\epsilon\}\cup\{x_{00} + x_{01} - n\theta_{00} - n\theta_{01} >
n\epsilon\}$ occurs. Then, using the Markov inequality,
\footnotesize
\begin{align*}
P(\mathbf{x} \in \overline{R}_n^c) &\leq P(\{n\theta_{00} - x_{00} > n \epsilon\}\cup \{x_{00} + x_{10} - n\theta_{00} - n\theta_{10} >
n\epsilon\}\\ &\qquad \qquad \cup\{x_{00} + x_{01} - n\theta_{00} - n\theta_{01} >
n\epsilon\}) \\ &\leq P(|x_{00} -n\theta_{00}| > n\epsilon) + P(|x_{00} + x_{10} - n\theta_{00} - n\theta_{10}| >
									n\epsilon)\\ &\qquad \qquad + P(|x_{00} + x_{00} - n\theta_{00} - n\theta_{01}| > n\epsilon) \\
									&< \frac{\mathrm{E}(x_{00} -n\theta_{00})^4 + \mathrm{E}(x_{00} + x_{10} - n\theta_{00} - n\theta_{10})^4 +
											\mathrm{E}(x_{00} + x_{00} - n\theta_{00} - n\theta_{01})^4}{n^4\epsilon^4}
									\\ &= \frac{C}{n^2},
\end{align*}
                for some constant $C$ which does not depend on $n$.
\end{proof}

\begin{lemma}
  \label{lm:g1}
Let $\mathbf{f} = (f_1, f_2, f_3) = \left(\left(\frac{x_{00} + x_{10}}{n}\right)^{1 / k}, \left(\frac{x_{00} + x_{01}}{n}\right)^{1 / k}, \left(\frac{x_{00}}{n}\right)^{1 / k}\right)$, then the elements of $\mathbf{f}$ can be expressed as follows:
\begin{itemize}
	\item[(a)] $f_1(x_{00} + x_{10}) = p_{00} + p_{10} + \frac{1}{nk(p_{00} + p_{10})^{k - 1}}(x_{00} + x_{10} - n(p_{00} + p_{10})^k) +
							\frac{1 - k}{2n^2k^2(p_{00} + p_{10})^{2k - 1}}(x_{00} + x_{10} - n(p_{00} + p_{10})^k)^2 + O_p(n^{-2});$
	\item[(b)] $f_2(x_{00} + x_{01}) = p_{00} + p_{01} + \frac{1}{nk(p_{00} + p_{01})^{k - 1}}(x_{00} + x_{01} - n(p_{00} + p_{01})^k) +
							\frac{1 - k}{2n^2k^2(p_{00} + p_{01})^{2k - 1}}(x_{00} + x_{01} - n(p_{00} + p_{01})^k)^2 + O_p(n^{-2});$
	\item[(c)] $f_3(x_{00}) = p_{00} + \frac{1}{nkp_{00}^{k - 1}}(x_{00} - np_{00}^k) +
							\frac{1 - k}{2n^2k^2 p_{00}^{2k - 1}}(x_{00} - np_{00}^k)^2 + O_p(n^{-2}).$
\end{itemize}
\end{lemma}

\begin{proof}
For each element of $f$, the component random variables can be expressed as a single marginal binomial random variable (e.g., $x_{00} + x_{10} \sim Bin(n, (p_{00} + p_{10})^k)$). The result then follows by taking the second order Taylor expansion of each element about the mean of the constituent random variable.

We illustrate this for $f_1$. Let $\tilde{x} = x_{00} + x_{10}$ and set $\tilde{x}_0 = \mathrm{E}(\tilde{x}) = n(p_{00} + p_{10})^k.$ Then, since $\frac{df_1}{d\tilde{x}} = \frac{\xi}{n}\left(\frac{\tilde{x}}{n}\right)^{\xi -1 }$ and 
$\frac{d^2f_1}{d\tilde{x}^2} = \frac{\xi(\xi - 1)}{n^2}\left(\frac{\tilde{x}}{n}\right)^{\xi -2}$, the Taylor expansion of $f_1$ about $\tilde{x}_0$ yields, letting $\xi = 1 / k$,
\begin{align*}
	f_1 &= \left(\frac{\tilde{x}_0}{n}\right)^{\xi} + 
	\frac{\xi}{n}\left(\frac{\tilde{x}_0}{n}\right)^{\xi-1}(\tilde{x} - \tilde{x}_0) + 
	\frac{\xi(\xi - 1)}{2n^2}\left(\frac{\tilde{x}_0}{n}\right)^{\xi-2}(\tilde{x} - \tilde{x}_0)^2 + O_p(n^{-2})\\
	&=p_{00} + p_{10} + \frac{1}{nk(p_{00} + p_{10})^{k - 1}}(x_{00} + x_{10} - n(p_{00} + p_{10})^k) \\
	&\qquad +	\frac{1 - k}{2n^2k^2(p_{00} + p_{10})^{2k - 1}}(x_{00} + x_{10} - n(p_{00} + p_{10})^k)^2 + O_p(n^{-2}).
\end{align*}
\end{proof}

\begin{lemma}
  \label{lm:g2}
  Let
  $\mathbf{g} = (g_1, g_2, g_3) = \left(\left(\frac{x_{00} + x_{10} + \eta}{n + \eta}\right)^{1 /
      k}, \left(\frac{x_{00} + x_{01} + \eta}{n + \eta}\right)^{1 / k}, \left(\frac{x_{00} + \eta}{n
        + \eta}\right)^{1 / k}\right)$, then the elements of $\mathbf{g}$ can be expressed as
  follows:
\begin{itemize}
	\item[(a)] $g_1(x_{00} + x_{10}) = p_{00} + p_{10} + \frac{1 - (p_{00} + p_{10})^k}{nk(p_{00} + p_{10})^{k - 1}}\eta + \frac{1}{nk(p_{00} + p_{10})^{k - 1}}(x_{00} + x_{10} - n(p_{00} + p_{10})^k) +
							\frac{1 - k}{2n^2k^2(p_{00} + p_{10})^{2k - 1}}(x_{00} + x_{10} - n(p_{00} + p_{10})^k)^2 + O_p(n^{-2});$
	\item[(b)] $g_2(x_{00} + x_{01}) = p_{00} + p_{01} + \frac{1 - (p_{00} + p_{01})^k}{nk(p_{00} + p_{01})^{k - 1}}\eta + \frac{1}{nk(p_{00} + p_{01})^{k - 1}}(x_{00} + x_{01} - n(p_{00} + p_{01})^k) +
							\frac{1 - k}{2n^2k^2(p_{00} + p_{01})^{2k - 1}}(x_{00} + x_{01} - n(p_{00} + p_{01})^k)^2 + O_p(n^{-2});$
	\item[(c)] $g_3(x_{00}) = p_{00} + \frac{1 - p_{00}^k}{nkp_{00}^{k - 1}}\eta+ \frac{1}{nkp_{00}^{k - 1}}(x_{00} - np_{00}^k) +
							\frac{1 - k}{2n^2k^2 p_{00}^{2k - 1}}(x_{00} - np_{00}^k)^2 + O_p(n^{-2}).$
\end{itemize}
\end{lemma}

\begin{proof}
The proof proceeds identically as for the previous lemma by first finding the second order taylor expansion of the elements of $\mathbf{g}$ about the mean of the respective component random variables. The result then follows by finding the Taylor expansion of each resultant term about $\eta = 0$.

We illustrate this for $g_1$, for which, letting $\tilde{x}$, $\tilde{x}_0$, and $\xi$ be as defined in the proof of the previous lemma, has the Taylor expansion about $\tilde{x}_0$
\begin{equation}
	\label{tg1}
	g_1 = \left(\frac{\tilde{x}_0 + \eta}{n + \eta}\right)^{\xi} + \frac{\xi(\tilde{x}_0 + \eta)^{\xi - 1}}{(n + \eta)^{\xi}}(\tilde{x} - \tilde{x}_0) + \frac{\xi(\xi - 1)(\tilde{x}_0 + \eta)^{\xi - 2}}{2(n + \eta)^{\xi}}(\tilde{x} - \tilde{x}_0)^2 + O_p(n^{-2}).
\end{equation}
For an integer $r$, we have \[\frac{d}{d\eta}\left(\frac{(\tilde{x}_0 + \eta)^{\xi - r}}{(n + \eta)^\xi}\right) =
\frac{\xi - r}{(n + \eta)^{r + 1}}\left(\frac{\tilde{x}_0 + \eta}{n + \eta}\right)^{\xi - r - 1} - \frac{\xi}{(n + \eta)^{r + 1}} \left(\frac{\tilde{x}_0 + \eta}{n + \eta}\right)^{\xi - r},\] so that for $r > 1$ this derivative is $O(n^{-3})$. Likewise, for all $r \geq 1$ the second derivative of the same term will be $O(n^{-3})$. As such, taking the Taylor expansion of (\ref{tg1}) about $\eta = 0$ yields
{\small
\begin{align*}
	g_1 &= \left(\frac{\tilde{x}_0}{n}\right)^{\xi} + \frac{\xi}{n}\left(\frac{\tilde{x}_0}{n}\right)^{\xi - 1}\eta - 
	\frac{\xi}{n}\left(\frac{\tilde{x}_0}{n}\right)^{\xi} \eta + \frac{\xi}{x}\left(\frac{\tilde{x}_0}{n}\right)^{\xi - 1} (\tilde{x} - \tilde{x}_0)\\ 
	&\qquad + \frac{\xi(\xi - 1)}{2n^2}\left(\frac{\tilde{x}_0}{n}\right)^{\xi - 2}(\tilde{x} - \tilde{x}_0)^2 + O_p(n^{-2})\\
	&= p_{00} + p_{10} + \frac{1 - (p_{00} + p_{10})^k}{nk(p_{00} + p_{10})^{k - 1}}\eta + \frac{1}{nk(p_{00} + p_{10})^{k - 1}}(x_{00} + x_{10} - n(p_{00} + p_{10})^k)\\ &\qquad +
							\frac{1 - k}{2n^2k^2(p_{00} + p_{10})^{2k - 1}}(x_{00} + x_{10} - n(p_{00} + p_{10})^k)^2 + O_p(n^{-2})
\end{align*}}
\end{proof}

\subsection{Proof of Theorem \ref{thm:distribution}}

(a) By the definition of the $RMM$ estimator, we have $\mathbf{p}^{RMM} - \mathbf{p}^{MLE} = h(\mathbf{x})I(\mathbf{x} \in \overline{R}_n^c)$,
where $h$ is some bounded function and $I$ is the indicator function. Then, from (b) in Theorem \ref{md:mle1} the right hand side of this expression converges
to zero almost surely, and the result follows.

(b) Each estimator can be expressed as $\hat{\mathbf{p}} = h_1(\mathbf{x}) +
h_2(\mathbf{x})I(\mathbf{x} \in \overline{R}_n^c)$ for some bounded  functions $h_1, h_2$. By standard
multinomial theory, $\mathbf{x}$ is asymptotically normal with rate $O(n^{-1})$ so that the first term, $h_1$, is as well
by the Delta method.  By Lemma \ref{lm:prn} and the boundedness of $h_2$, the second term is of
order $O(n^{-2})$ so that the overall convergence (up to $O(n^{-1})$) is determined only by the first
term. The exact values of the asymptotic covariance matrix can then be calculated directly using the
Taylor expansions given in Lemmas \ref{lm:g1} and \ref{lm:g2}. Some algebra shows that each
estimator yields an identical covariance matrix, so that the large sample
distributions, up to $O(n^{-1})$, are identical.

We illustrate the calculations for the covariance matrix (as given in Appendix \ref{sec:apSigma}) by finding $\boldsymbol{\Sigma}_{11}$, where $\frac{1}{nk^2}\boldsymbol{\Sigma}_{11}$ is the asymptotic variance of $p_{10}$  for the MLE. Note that, for $p_{11}^{MLE}$, the function $h_1$ defined above is given by $f_1 - f_3$, where $f_1, f_3$ are as in Lemma \ref{lm:g1}. As such, ignoring terms of $O(n^{-2})$, the asymptotic variance of $p_{10}^{MLE}$ is given by, 
defining $\lambda_{10} = p_{00} + p_{10}$, 
{\small
\begin{align*}
	\boldsymbol{\Sigma}_{11} = nk^2\mathrm{E}\left(f_1 - f_3 - p_{10}\right)^2 &= \frac{1}{n}\mathrm{E}\left(\frac{(x_{00} + x_{10} - n\lambda_{10}^k)}{\lambda_{10}^{k-1}} - \frac{(x_{00}- n\lambda_{00}^k)}{p_{00}^{k-1}}\right)^2 \\
	&= \frac{1}{n}\left(\frac{\mathrm{V}(x_{00} + x_{10})}{\lambda_{10}^{2k - 2}} + \frac{\mathrm{V}(x_{00})}{p_{00}^{2k - 2}} - 2 \frac{\mathrm{Cov}(x_{00} + x_{10}, x_{00})}{\lambda_{10}^{k-1}p_{00}^{k-1}}\right)\\
	&= \frac{1}{n}\left(\frac{n\lambda_{10}^k(1 - \lambda_{10}^k)}{\lambda_{10}^{2k - 2}} + \frac{np_{00}^k(1 - p_{00}^k)}{p_{00}^{2k - 2}}\right.\\
	&\qquad \qquad \left.\ - 2 \frac{np_{00}^k(1 - p_{00}^k) - np_{00}^k(\lambda_{10}^k - p_{00}^k)}{\lambda_{10}^{k-1}p_{00}^{k-1}}\right)\\
	&= \frac{1}{\lambda_{10}^k}\left[ \lambda_{10}^2 - 2 \lambda_{10}p_{00}\right] + \frac{p_{00}^2}{p_{00}^k} - (\lambda_{10} - p_{00})^2\\
	&= p_{10}^2\left(\frac{1}{\lambda_{10}^k} - 1\right) + p_{00}^2\left(\frac{1}{p_{00}^k} - \frac{1}{\lambda_{10}^k}\right).
\end{align*}}

Similar calculations yield the other values of $\boldsymbol{\Sigma}$, noting that the $h_1$ functions for $p_{01}^{MLE}$ and $p_{11}^{MLE}$ are given by $f_2 - f_3$ and $1 - f_1 - f_2 + f_3$, respectively.

To see that the Burrows estimator yields the same asymptotic covariance matrix, note that the function $h_1$ defined above for $p_{10}^{B}$ is $g_1 - g_3$, where $g_1$ and $g_3$ are as defined in Lemma \ref{lm:g2}. Now, for any value of $\eta$, we have, ignoring terms of $O(n^{-2})$, $g_1 - g_3 = f_1 - f_3 + C,$ where $C$ is constant with respect to $\mathbf{x}$ and of order $O(n^{-1})$. As such, the asymptotic variance of $p_{10}^{B}$ is given by $\mathrm{E}(f_1 - f_3 + C - p_{10})^2 = \frac{1}{nk^2}\boldsymbol{\Sigma}_{11} + 2C\mathrm{E}(f_1 - f_3 - p_{10}) + C^2 = \frac{1}{nk^2}\boldsymbol{\Sigma}_{11} + O(n^{-2})$. Similar reasoning holds for the other elements of the asymptotic covariance matrix for the Burrows estimator.

\subsection{Proof of Theorem \ref{thm:bias}}
This proof is nearly identical to (b) in the previous theorem. The first order expectations can be
found using the Taylor expansions in Lemmas \ref{lm:g1} and \ref{lm:g2}, and combined as in the proof the the previous theorem.

\section{Additional Tables}
\label{sec:tables} 
\begin{table}[H]
	\caption{Relative bias, defined for the $i^{th}$ element to be $100 \times \frac{\mathrm{E}(\hat{p}_i - p_i)}{p_i}$, for $k = 2$}
	\label{tab:bias1}
	\centering
	\resizebox{\columnwidth}{!}{%
		\begin{tabular}{llccccccccc}
			\toprule
			\multicolumn{2}{r}{$(p_{10}\ p_{01}\ p_{11}) = $} & (0.001 & 0.001 & 0.0001) & (0.045 & 0.045 & 0.005) & (0.095 & 0.045 & 0.005) \\ 
			\cmidrule(lr){3-5} \cmidrule(lr){6-8} \cmidrule(lr){9-11}
			$n = 10$ & $\hat{\mathbf{p}}^{MLE}$  & 2.536 & 2.536 & 3.658 & -1.089 & -1.089 & 38.991 & -0.726 & -5.012 & 76.072 \\   
			& $\hat{\mathbf{p}}^{RMM}$           & 2.533 & 2.533 & 3.658 & -1.312 & -1.312 & 38.991 & -1.040 & -5.432 & 76.072 \\   
			& $\hat{\mathbf{p}}^{B}$             & -0.036 & -0.036 & 1.063 & -3.889 & -3.889 & 35.701 & -3.719 & -7.947 & 72.229 \\ 
			$n = 25$ & $\hat{\mathbf{p}}^{MLE}$ & 0.934 & 0.934 & 2.026 & -2.038 & -2.038 & 29.988 & -1.598 & -4.819 & 55.740 \\   
			& $\hat{\mathbf{p}}^{RMM}$          & 0.921 & 0.921 & 2.026 & -2.166 & -2.166 & 29.988 & -1.795 & -5.027 & 55.740 \\   
			& $\hat{\mathbf{p}}^{B}$            & -0.089 & -0.089 & 1.005 & -3.192 & -3.192 & 28.838 & -2.862 & -6.056 & 54.610 \\ 
			$n = 50$ & $\hat{\mathbf{p}}^{MLE}$ & 0.420 & 0.420 & 1.500 & -1.656 & -1.656 & 20.740 & -1.167 & -3.207 & 35.042 \\   
			& $\hat{\mathbf{p}}^{RMM}$          & 0.406 & 0.406 & 1.500 & -1.728 & -1.728 & 20.740 & -1.282 & -3.317 & 35.042 \\   
			& $\hat{\mathbf{p}}^{B}$            & -0.097 & -0.097 & 0.993 & -2.248 & -2.248 & 20.259 & -1.822 & -3.851 & 34.685 \\ 
			$n = 100$ & $\hat{\mathbf{p}}^{MLE}$ & 0.168 & 0.168 & 1.234 & -0.819 & -0.819 & 10.236 & -0.521 & -1.469 & 16.252 \\   
			& $\hat{\mathbf{p}}^{RMM}$          & 0.153 & 0.153 & 1.234 & -0.851 & -0.851 & 10.236 & -0.576 & -1.519 & 16.252 \\   
			& $\hat{\mathbf{p}}^{B}$            & -0.097 & -0.097 & 0.981 & -1.116 & -1.116 & 10.049 & -0.849 & -1.794 & 16.153 \\ 
			\cmidrule(lr){3-11}
			\multicolumn{2}{r}{$(p_{10}\ p_{01}\ p_{11}) = $} & (0.1 & 0.1 & 0.1) & (0.15 & 0.1 & 0.2) & (0.25 & 0.05 & 0.15) \\ 
			\cmidrule(lr){3-5} \cmidrule(lr){6-8} \cmidrule(lr){9-11}
			$n = 10$ & $\hat{\mathbf{p}}^{MLE}$  & 2.808 & 2.808 & 3.457 & 7.017 & 7.664 & 1.417 & 5.662 & 6.391 & 2.047 \\
			& $\hat{\mathbf{p}}^{RMM}$           & 2.670 & 2.670 & 3.457 & 6.911 & 7.536 & 1.417 & 5.527 & 6.112 & 2.047 \\
			& $\hat{\mathbf{p}}^{B}$             & -1.105 & -1.105 & 1.309 & 0.327 & 0.299 & 0.068 & 0.156 & -0.573 & 0.326 \\
			$n = 25$ & $\hat{\mathbf{p}}^{MLE}$ & 1.415 & 1.415 & 0.911 & 2.009 & 2.147 & 0.824 & 1.789 & 2.281 & 0.788 \\
			& $\hat{\mathbf{p}}^{RMM}$          & 1.411 & 1.411 & 0.911 & 2.008 & 2.146 & 0.824 & 1.787 & 2.278 & 0.788 \\
			& $\hat{\mathbf{p}}^{B}$            & -0.031 & -0.031 & 0.061 & 0.068 & 0.079 & -0.010 & 0.047 & 0.064 & -0.001 \\
			$n = 50$ & $\hat{\mathbf{p}}^{MLE}$ & 0.715 & 0.715 & 0.428 & 0.944 & 1.002 & 0.428 & 0.853 & 1.070 & 0.405 \\
			& $\hat{\mathbf{p}}^{RMM}$          & 0.715 & 0.715 & 0.428 & 0.944 & 1.002 & 0.428 & 0.853 & 1.070 & 0.405 \\
			& $\hat{\mathbf{p}}^{B}$            & 0.006 & 0.006 & 0.001 & 0.015 & 0.017 & -0.001 & 0.012 & 0.020 & -0.002 \\
			$n = 100$ & $\hat{\mathbf{p}}^{MLE}$ & 0.353 & 0.353 & 0.214 & 0.460 & 0.487 & 0.216 & 0.417 & 0.518 & 0.205 \\
			& $\hat{\mathbf{p}}^{RMM}$          & 0.353 & 0.353 & 0.214 & 0.460 & 0.487 & 0.216 & 0.417 & 0.518 & 0.205 \\
			& $\hat{\mathbf{p}}^{B}$            & 0.002 & 0.002 & 0.000 & 0.004 & 0.004 & 0.000 & 0.003 & 0.005 & 0.000 \\
			\bottomrule
			
	\end{tabular}}
\end{table}

\begin{table}[H]
	\caption{Relative bias, defined for the $i^{th}$ element to be $100 \times \frac{\mathrm{E}(\hat{p}_i - p_i)}{p_i}$, for $k = 10$}
	\label{tab:bias2}
	\centering
	\resizebox{\columnwidth}{!}{%
		\begin{tabular}{llccccccccc}
			\toprule
			\multicolumn{2}{r}{$(p_{10}\ p_{01}\ p_{11}) = $} & (0.001 & 0.001 & 0.0001) & (0.045 & 0.045 & 0.005) & (0.095 & 0.045 & 0.005) \\ 
			\cmidrule(lr){3-5} \cmidrule(lr){6-8} \cmidrule(lr){9-11}
			$n = 10$ & $\hat{\mathbf{p}}^{MLE}$  & 3.922 & 3.922 & 14.062 & -4.274 & -4.274 & 108.752 & 9.787 & -14.963 & 206.077 \\
			& $\hat{\mathbf{p}}^{RMM}$           & 3.914 & 3.914 & 14.062 & -4.623 & -4.623 & 108.752 & 9.205 & -15.437 & 206.077 \\
			& $\hat{\mathbf{p}}^{B}$             & -0.800 & -0.800 & 8.924 & -12.035 & -12.035 & 109.953 & -10.672 & -23.050 & 209.083 \\
			$n = 25$ & $\hat{\mathbf{p}}^{MLE}$ & 0.993 & 0.993 & 10.694 & -2.916 & -2.916 & 51.631 & -1.630 & -8.554 & 102.982 \\
			& $\hat{\mathbf{p}}^{RMM}$          & 0.973 & 0.973 & 10.694 & -3.065 & -3.065 & 51.631 & -1.904 & -8.770 & 102.982 \\
			& $\hat{\mathbf{p}}^{B}$            & -0.856 & -0.856 & 8.705 & -6.106 & -6.106 & 55.188 & -5.949 & -12.646 & 114.048 \\
			$n = 50$ & $\hat{\mathbf{p}}^{MLE}$ & 0.088 & 0.088 & 9.442 & -1.646 & -1.646 & 27.247 & -1.402 & -5.377 & 61.226 \\
			& $\hat{\mathbf{p}}^{RMM}$          & 0.064 & 0.064 & 9.442 & -1.724 & -1.724 & 27.247 & -1.559 & -5.500 & 61.226 \\
			& $\hat{\mathbf{p}}^{B}$            & -0.842 & -0.842 & 8.461 & -3.301 & -3.301 & 29.770 & -3.601 & -7.621 & 68.648 \\
			$n = 100$ & $\hat{\mathbf{p}}^{MLE}$ & -0.330 & -0.330 & 8.494 & -0.685 & -0.685 & 12.290 & -0.887 & -3.062 & 33.958 \\
			& $\hat{\mathbf{p}}^{RMM}$          & -0.349 & -0.349 & 8.494 & -0.722 & -0.722 & 12.290 & -0.974 & -3.129 & 33.958 \\
			& $\hat{\mathbf{p}}^{B}$            & -0.800 & -0.800 & 8.010 & -1.549 & -1.549 & 13.957 & -2.030 & -4.290 & 38.628 \\
			\cmidrule(lr){3-11}
			\multicolumn{2}{r}{$(p_{10}\ p_{01}\ p_{11}) = $} & (0.1 & 0.1 & 0.1) & (0.15 & 0.1 & 0.2) & (0.25 & 0.05 & 0.15) \\ 
			\cmidrule(lr){3-5} \cmidrule(lr){6-8} \cmidrule(lr){9-11}
			$n = 10$ & $\hat{\mathbf{p}}^{MLE}$  & 111.716 & 111.716 & 129.351 & 11.390 & -31.535 & 265.909 & 110.522 & -72.403 & 184.599 \\
			& $\hat{\mathbf{p}}^{RMM}$           & 111.318 & 111.318 & 129.351 & 11.368 & -31.566 & 265.909 & 110.506 & -72.460 & 184.599 \\
			& $\hat{\mathbf{p}}^{B}$             & -33.310 & -33.310 & 25.769 & -84.340 & -88.696 & 17.266 & -69.376 & -90.862 & 25.260 \\
			$n = 25$ & $\hat{\mathbf{p}}^{MLE}$ & 59.336 & 59.336 & -7.420 & 90.400 & 16.338 & 151.864 & 166.329 & -60.572 & 54.524 \\
			& $\hat{\mathbf{p}}^{RMM}$          & 58.873 & 58.873 & -7.420 & 90.335 & 16.252 & 151.864 & 166.284 & -60.668 & 54.524 \\
			& $\hat{\mathbf{p}}^{B}$            & 12.736 & 12.736 & -12.717 & -59.829 & -63.740 & 23.014 & -45.998 & -69.863 & 23.300 \\
			$n = 50$ & $\hat{\mathbf{p}}^{MLE}$ & 24.030 & 24.030 & -14.632 & 135.538 & 52.818 & 52.866 & 157.344 & -37.670 & 18.687 \\
			& $\hat{\mathbf{p}}^{RMM}$          & 23.786 & 23.786 & -14.632 & 135.549 & 52.844 & 52.866 & 157.268 & -37.754 & 18.687 \\
			& $\hat{\mathbf{p}}^{B}$            & 10.603 & 10.603 & -10.444 & -19.383 & -14.790 & 5.401 & -26.463 & -40.935 & 13.750 \\
			$n = 100$ & $\hat{\mathbf{p}}^{MLE}$ & 10.090 & 10.090 & -6.788 & 131.223 & 82.807 & -19.231 & 121.669 & 1.298 & 1.744 \\
			& $\hat{\mathbf{p}}^{RMM}$          & 10.033 & 10.033 & -6.788 & 131.371 & 82.942 & -19.231 & 121.667 & 1.259 & 1.744 \\
			& $\hat{\mathbf{p}}^{B}$            & 1.437 & 1.437 & -1.403 & 32.649 & 53.351 & -26.689 & -8.370 & -1.585 & 0.551 \\
			\bottomrule
			
	\end{tabular}}
\end{table}

\begin{table}[H]
	\caption{$1000\ \times$ mean squared error (MSE) for $k = 2$}
	\label{tab:mse1}
	\centering
	\resizebox{\columnwidth}{!}{%
		\begin{tabular}{llccccccccc}
			\toprule
			\multicolumn{2}{r}{$(p_{10}\ p_{01}\ p_{11}) = $} & (0.001 & 0.001 & 0.0001) & (0.045 & 0.045 & 0.005) & (0.095 & 0.045 & 0.005) \\ 
			\cmidrule(lr){3-5} \cmidrule(lr){6-8} \cmidrule(lr){9-11}
			$n = 10$ & $\hat{\mathbf{p}}^{MLE}$  & 0.053 & 0.053 & 0.005 & 2.286 & 2.286 & 0.351 & 4.779 & 2.248 & 0.446 \\
			& $\hat{\mathbf{p}}^{RMM}$           & 0.053 & 0.053 & 0.005 & 2.272 & 2.272 & 0.351 & 4.734 & 2.222 & 0.446 \\
			& $\hat{\mathbf{p}}^{B}$             & 0.050 & 0.050 & 0.005 & 2.154 & 2.154 & 0.333 & 4.476 & 2.109 & 0.425 \\
			$n = 25$ & $\hat{\mathbf{p}}^{MLE}$ & 0.020 & 0.020 & 0.002 & 0.895 & 0.895 & 0.127 & 1.867 & 0.903 & 0.156 \\
			& $\hat{\mathbf{p}}^{RMM}$          & 0.020 & 0.020 & 0.002 & 0.892 & 0.892 & 0.127 & 1.856 & 0.898 & 0.156 \\
			& $\hat{\mathbf{p}}^{B}$            & 0.020 & 0.020 & 0.002 & 0.874 & 0.874 & 0.125 & 1.818 & 0.880 & 0.154 \\
			$n = 50$ & $\hat{\mathbf{p}}^{MLE}$ & 0.010 & 0.010 & 0.001 & 0.451 & 0.451 & 0.060 & 0.938 & 0.465 & 0.074 \\
			& $\hat{\mathbf{p}}^{RMM}$          & 0.010 & 0.010 & 0.001 & 0.450 & 0.450 & 0.060 & 0.935 & 0.464 & 0.074 \\
			& $\hat{\mathbf{p}}^{B}$            & 0.010 & 0.010 & 0.001 & 0.445 & 0.445 & 0.059 & 0.926 & 0.459 & 0.073 \\
			$n = 100$ & $\hat{\mathbf{p}}^{MLE}$ & 0.005 & 0.005 & 0.001 & 0.229 & 0.229 & 0.030 & 0.471 & 0.237 & 0.038 \\
			& $\hat{\mathbf{p}}^{RMM}$          & 0.005 & 0.005 & 0.001 & 0.229 & 0.229 & 0.030 & 0.470 & 0.236 & 0.038 \\
			& $\hat{\mathbf{p}}^{B}$            & 0.005 & 0.005 & 0.001 & 0.227 & 0.227 & 0.030 & 0.467 & 0.235 & 0.038 \\
			\cmidrule(lr){3-11}
			\multicolumn{2}{r}{$(p_{10}\ p_{01}\ p_{11}) = $} & (0.1 & 0.1 & 0.1) & (0.15 & 0.1 & 0.2) & (0.25 & 0.05 & 0.15) \\ 
			\cmidrule(lr){3-5} \cmidrule(lr){6-8} \cmidrule(lr){9-11}
			$n = 10$ & $\hat{\mathbf{p}}^{MLE}$  & 6.581 & 6.581 & 5.672 & 13.421 & 9.927 & 11.993 & 18.174 & 5.241 & 8.946 \\
			& $\hat{\mathbf{p}}^{RMM}$           & 6.546 & 6.546 & 5.672 & 13.331 & 9.860 & 11.993 & 17.994 & 5.188 & 8.946 \\
			& $\hat{\mathbf{p}}^{B}$             & 6.001 & 6.001 & 5.345 & 10.794 & 7.940 & 11.017 & 14.779 & 4.368 & 8.338 \\
			$n = 25$ & $\hat{\mathbf{p}}^{MLE}$ & 2.494 & 2.494 & 2.268 & 4.349 & 3.196 & 4.444 & 6.008 & 1.766 & 3.426 \\
			& $\hat{\mathbf{p}}^{RMM}$          & 2.494 & 2.494 & 2.268 & 4.349 & 3.196 & 4.444 & 6.006 & 1.766 & 3.426 \\
			& $\hat{\mathbf{p}}^{B}$            & 2.416 & 2.416 & 2.215 & 4.145 & 3.045 & 4.316 & 5.722 & 1.682 & 3.332 \\
			$n = 50$ & $\hat{\mathbf{p}}^{MLE}$ & 1.224 & 1.224 & 1.123 & 2.094 & 1.537 & 2.173 & 2.898 & 0.849 & 1.680 \\
			& $\hat{\mathbf{p}}^{RMM}$          & 1.224 & 1.224 & 1.123 & 2.094 & 1.537 & 2.173 & 2.898 & 0.849 & 1.680 \\
			& $\hat{\mathbf{p}}^{B}$            & 1.205 & 1.205 & 1.110 & 2.049 & 1.503 & 2.143 & 2.834 & 0.830 & 1.657 \\
			$n = 100$ & $\hat{\mathbf{p}}^{MLE}$ & 0.606 & 0.606 & 0.557 & 1.030 & 0.756 & 1.075 & 1.427 & 0.417 & 0.832 \\
			& $\hat{\mathbf{p}}^{RMM}$          & 0.606 & 0.606 & 0.557 & 1.030 & 0.756 & 1.075 & 1.427 & 0.417 & 0.832 \\
			& $\hat{\mathbf{p}}^{B}$            & 0.601 & 0.601 & 0.554 & 1.019 & 0.748 & 1.068 & 1.411 & 0.413 & 0.826 \\
			\bottomrule
			
	\end{tabular}}
\end{table}

\begin{table}[H]
	\caption{$1000 \ \times$ mean squared error (MSE) for $k = 10$}
	\label{tab:mse2}
	\centering
	\resizebox{\columnwidth}{!}{%
		\begin{tabular}{llccccccccc}
			\toprule
			\multicolumn{2}{r}{$(p_{10}\ p_{01}\ p_{11}) = $} & (0.001 & 0.001 & 0.0001) & (0.045 & 0.045 & 0.005) & (0.095 & 0.045 & 0.005) \\ 
			\cmidrule(lr){3-5} \cmidrule(lr){6-8} \cmidrule(lr){9-11}
			$n = 10$ & $\hat{\mathbf{p}}^{MLE}$  & 0.011 & 0.011 & 0.001 & 0.848 & 0.848 & 0.225 & 11.757 & 1.044 & 0.526 \\
			& $\hat{\mathbf{p}}^{RMM}$           & 0.011 & 0.011 & 0.001 & 0.838 & 0.838 & 0.225 & 11.710 & 1.030 & 0.526 \\
			& $\hat{\mathbf{p}}^{B}$             & 0.010 & 0.010 & 0.001 & 0.654 & 0.654 & 0.210 & 1.679 & 0.859 & 0.492 \\
			$n = 25$ & $\hat{\mathbf{p}}^{MLE}$ & 0.004 & 0.004 & 0.000 & 0.288 & 0.288 & 0.085 & 0.782 & 0.368 & 0.185 \\
			& $\hat{\mathbf{p}}^{RMM}$          & 0.004 & 0.004 & 0.000 & 0.286 & 0.286 & 0.085 & 0.772 & 0.365 & 0.185 \\
			& $\hat{\mathbf{p}}^{B}$            & 0.004 & 0.004 & 0.000 & 0.272 & 0.272 & 0.084 & 0.703 & 0.359 & 0.190 \\
			$n = 50$ & $\hat{\mathbf{p}}^{MLE}$ & 0.002 & 0.002 & 0.000 & 0.146 & 0.146 & 0.045 & 0.374 & 0.189 & 0.095 \\
			& $\hat{\mathbf{p}}^{RMM}$          & 0.002 & 0.002 & 0.000 & 0.145 & 0.145 & 0.045 & 0.370 & 0.188 & 0.095 \\
			& $\hat{\mathbf{p}}^{B}$            & 0.002 & 0.002 & 0.000 & 0.142 & 0.142 & 0.045 & 0.361 & 0.189 & 0.099 \\
			$n = 100$ & $\hat{\mathbf{p}}^{MLE}$ & 0.001 & 0.001 & 0.000 & 0.076 & 0.076 & 0.025 & 0.190 & 0.100 & 0.052 \\
			& $\hat{\mathbf{p}}^{RMM}$          & 0.001 & 0.001 & 0.000 & 0.075 & 0.075 & 0.025 & 0.188 & 0.100 & 0.052 \\
			& $\hat{\mathbf{p}}^{B}$            & 0.001 & 0.001 & 0.000 & 0.075 & 0.075 & 0.025 & 0.187 & 0.100 & 0.053 \\
			\cmidrule(lr){3-11}
			\multicolumn{2}{r}{$(p_{10}\ p_{01}\ p_{11}) = $} & (0.1 & 0.1 & 0.1) & (0.15 & 0.1 & 0.2) & (0.25 & 0.05 & 0.15) \\ 
			\cmidrule(lr){3-5} \cmidrule(lr){6-8} \cmidrule(lr){9-11}
			$n = 10$ & $\hat{\mathbf{p}}^{MLE}$  & 107.093 & 107.093 & 108.999 & 103.465 & 48.242 & 430.127 & 229.936 & 9.708 & 224.885 \\
			& $\hat{\mathbf{p}}^{RMM}$           & 107.021 & 107.021 & 108.999 & 103.462 & 48.235 & 430.127 & 229.939 & 9.700 & 224.885 \\
			& $\hat{\mathbf{p}}^{B}$             & 5.963 & 5.963 & 8.182 & 18.051 & 9.102 & 4.409 & 33.645 & 2.705 & 5.507 \\
			$n = 25$ & $\hat{\mathbf{p}}^{MLE}$ & 33.024 & 33.024 & 11.586 & 134.286 & 58.504 & 242.080 & 245.355 & 6.821 & 44.190 \\
			& $\hat{\mathbf{p}}^{RMM}$          & 32.906 & 32.906 & 11.586 & 134.262 & 58.476 & 242.080 & 245.350 & 6.805 & 44.190 \\
			& $\hat{\mathbf{p}}^{B}$            & 6.993 & 6.993 & 6.531 & 14.150 & 9.705 & 11.590 & 16.595 & 3.947 & 6.101 \\
			$n = 50$ & $\hat{\mathbf{p}}^{MLE}$ & 6.998 & 6.998 & 3.791 & 130.373 & 52.119 & 102.981 & 219.979 & 6.169 & 8.984 \\
			& $\hat{\mathbf{p}}^{RMM}$          & 6.934 & 6.934 & 3.791 & 130.382 & 52.136 & 102.981 & 219.952 & 6.156 & 8.984 \\
			& $\hat{\mathbf{p}}^{B}$            & 4.800 & 4.800 & 3.613 & 13.320 & 14.354 & 18.479 & 7.891 & 5.559 & 6.283 \\
			$n = 100$ & $\hat{\mathbf{p}}^{MLE}$ & 2.134 & 2.134 & 1.630 & 91.463 & 37.516 & 36.011 & 167.099 & 7.515 & 7.635 \\
			& $\hat{\mathbf{p}}^{RMM}$          & 2.120 & 2.120 & 1.630 & 91.560 & 37.585 & 36.011 & 167.102 & 7.510 & 7.635 \\
			& $\hat{\mathbf{p}}^{B}$            & 1.782 & 1.782 & 1.444 & 18.683 & 23.035 & 24.012 & 5.322 & 7.265 & 7.360 \\
			\bottomrule
			
	\end{tabular}}
\end{table}


\begin{thebibliography}{}
	
			\bibitem[Avrahami-Moyal et al.(2017)]{am2017}
                          Avrahami-Moyal, L., Tam, Y., Brumin, M., Prakash, S., Leibman, D., Pearlsman, M.,
                          Bornstein, M., Sela, N., Zeidan, M., Dar, Z., Zig, U., Gal-On, A., and Gaba, V. (2017).
                          Detection of Potato virus Y in industrial quantities of seed potatoes by TaqMan Real Time PCR.
			\textit{Phytoparasitica} \textbf{45} 591--598.

			\bibitem[\protect\citeauthoryear{Burrows}{1987}]{burrows1987}
			Burrows, P. M. (1987).
			Improved Estimation of Pathogen Transmission Rates by Group Testing.
			{\em Phytopathology} \textbf{77} 363--365.

			\bibitem[Ding and Xiong(2015)]{ding2015}
			Ding, J. and Xiong, W. (2015).
			Robust group testing for multiple traits with misclassification.
			\textit{Journal of Applied Statistics} \textbf{42} 2115--2125.
			
			\bibitem[Ding and Xiong(2016)]{ding2016}
			Ding, J. and Xiong, W. (2016).
			A new estimator for a population proportion using group testing.
			\textit{Communications in Statistics--Simulation and Computation} \textbf{45} 101--114.
	
			\bibitem[Fletcher(2012)]{fletcher2012}
Fletcher, J. D. (2012).
A virus survey of New Zealand fresh, process and seed potato crops during 2010-11.
\textit{New Zealand Plant Protection} \textbf{65} 197--203.

                        \bibitem[Gray et al.(2010)]{gray2010}
                        Gray, S., De Boer, S., Lorenzen, J., Karazev, A., Whitworth, J., Nolte, P., Singh, R.,
                        Boucher, A., and Xu, H. (2010).
                        Potato virus Y: an evolving concern for potato crops in the United States and Canada.
			{\em Plant Disease} \textbf{94} 1384--1397.


			\bibitem[Grend\'ar and \u Spitalsk\'y(2017)]{grend2017}
			Grend\'ar, M. and \u Spitalsk\'y, V. (2017).
			Multinomial and empirical likelihood under convex constraints: Directions of recession, Fenchel duality, the PP algorithm.
			\textit{Electronic Journal of Statistics} \textbf{11} 2547--2612.
			
			\bibitem[Haber and Malinovsky(2017)]{haber2017}
			Haber, G. and Malinovsky, Y. (2017).
			Random walk designs for selecting pool sizes in group testing estimation with small samples.
			\textit{Biometrical Journal} \textbf{59} 1382--1398.
			
			\bibitem[Haber and Malinovsky(2018)]{haberM2018}
			Haber, G. and Malinovsky, Y. (2018).
			On the construction of unbiased estimators for the group testing problem.
			\textit{Sankhya A}. https://doi.org/10.1007/s13171-018-0156-4.
			
			\bibitem[Haber et al.(2018)]{haber2018}
			Haber, G., Malinovsky, Y., and Albert, P. S. (2018).
			Sequential estimation in the group testing problem.
			\textit{Sequential Analysis} \textbf{37} 1--17.

			\bibitem[\protect\citeauthoryear{Hepworth and Watson}{2009}]{hepworth2009}
			Hepworth, G. and Watson, R. (2009).
			Debiased estimation of proportions in group testing.
			{\em Journal of Royal Statistical Society, Series C} \textbf{58} 105--121.

			\bibitem[Hughes-Oliver and Rosenberger(2000)]{ho2000}
			Hughes-Oliver, J. M. and Rosenberger, W. (2000).
			Efficient estimation of the prevalence of multiple rare traits.
			{\em Biometrika} \textbf{87} 315--327.
			
			\bibitem[Hughes-Oliver and Swallow(1994)]{hos1994}
			Hughes-Oliver, J. M. and Swallow, W. H. (1994).
			A two-stage adaptive group testing procedure for estimating small proportions.
			\textit{Journal of the American Statistical Association} \textbf{89} 982--993.
			
			\bibitem[Hyun et al.(2018)]{hyun2018}
			Hyun, N., Gastwirth, J. L., Graubard, B. I. (2018).
			Grouping methods for estimating prevalences of rare traits for complex survey data that preserve confidentiality of respondents.
			\textit{Statistics in Medicine} \textbf{37} 2174--2186.
			
			\bibitem[Jamshidian(2004)]{jam2004}
			Jamshidian, M. (2004).
			On algorithms for restricted maximum likelihood estimation.
			{\em Computational Statistics and Data Analysis} \textbf{45} 137--157.

			\bibitem[Li et al.(2017)]{li2017}
			Li, Q., Liu, A., and Xiong, W. (2017).
			D-Optimality of group testing for joint estimation of correlated rare diseases with misclassification.
			\textit{Statistica Sinica} \textbf{27} 823--838.
                        
                        \bibitem[Liu et al.(2011)]{liu2011}
			Liu, S. C., Chiang, K. S., Lin, C. H., Chung, W. C., Lin, S. H., and Yang, T. C. (2011).
                        Cost analysis in choosing group size when group testing for Potato virus Y in the presence
                        of classification errors.
			\textit{Annals of Applied Biology} \textbf{159} 491--502.
			    
			\bibitem[Liu et al.(2012)]{liu2012}
			Liu, A., Liu, C., Zhang, Z., and Albert, P. S. (2012).
			Optimality of group testing in the presence of misclassification.
			\textit{Biometrika} \textbf{99} 245--251.

                        \bibitem[Lorenzen(2006)]{lorenzen2006}
                          Lorenzen, J. H., Piche, L. M., Gudmestad, N. C., Meacham, T., and Shiel, P. (2006).
                          A multiplex PCR assay to characterize potato virus Y isolates and identify strain
                          mixtures. 
			\textit{Plant Disease} \textbf{90} 935--940.

                        \bibitem[Mallik et al.(2012)]{mallik2012}
			Mallik, I., Anderson, N. R., and Gudmestad, N. C. (2012).
                        Detection and differentiation of Potato Virus Y strains from potato using immunocapture
                        multiples RT-PCR.
			{\em American Journal of Potato Research} \textbf{89} 184--191.

                        \bibitem[Mello et al.(2011)]{mello2011}
                          Mello, A. F. S., Olarte, R. A., Gray, S. M., and Perry, K. L. (2011).
                          Transmission efficiency of Potato virus Y strains PVY\textsuperscript{O} and
                          PVY\textsuperscript{N-Wi} by five aphid species.
			\textit{Plant Disease} \textbf{95} 1279--1283.

                        \bibitem[Mondal et al.(2017)]{mondal2017}
                          Mondal, S., Lin, Y., Carroll, J. E., Wenninger, E. J., Bosque-Perez, N. A., Whitworth, J.
                          L., Hutchinson, P., Eigenbrode, S., and Gray, S. M. (2017).
                          Potato virus Y transmission efficiency from potato infected with single or multiple virus strains.
			{\em Phytopathology} \textbf{107} 491--498.
			
			\bibitem[Nelder and Mead(1965)]{nelder1965}
			Nelder, J. A. and Mead, R. (1965).
			A simplex method for function minimization.
			\textit{The Computer Journal} \textbf{7} 308--313.

			\bibitem[Nettleton(1999)]{net1999}
			Nettleton, D. (1999).
			Convergence properties of the EM Algorithm in constrained parameter spaces.
			{\em Canadian Journal of Statistics} \textbf{27} 639--648.

			\bibitem[Pfeiffer et al.(2002)]{pfeiffer2002}
			Pfeiffer, R. M., Rutter, J. L., Gail, M. H., Struewing, J., and Gastwirth, J. L. (2002).
			Efficiency of DNA pooling to estimate joint allele frequencies and measure linkage disequilibrium.
			\textit{Genetic Epidemiology} \textbf{22} 94--102.

			\bibitem[Santos and Dorgman(2016)]{santos2016}
			Santos, J. D. and Dorgman, D. (2016).
			An approximate likelihood estimator for the prevalence of infections in vectors using pools of varying sizes.
			\textit{Biometrical Journal} \textbf{58} 1248--1256.

			\bibitem[\protect\citeauthoryear{Swallow}{1985}]{swallow1985}
			Swallow, W. H. (1985).
			Group Testing for Estimating Infection Rates and Probabilities of Disease Transmission.
			{\em Phytopathology} \textbf{75} 882--889.

			\bibitem[Tebbs et al.(2003)]{tebbs2003}
			Tebbs, J. M., Bilder, C. R., and Koser, B. K. (2003).
			An empirical Bayes group-testing approach to estimating small proportions.
			\textit{Communications in Statistics -- Theory and Methods} \textbf{32} 983--995.

			\bibitem[Tebbs et al.(2013)]{tebbs2013}
			Tebbs, J. M., McMahan, C. S., and Bilder, C. R. (2013).
			Two-stage hierarchical group testing for multiple infections with application to the infertility prevention project.
			{\em Biometrics} \textbf{69} 1064--1073.

			\bibitem[\protect\citeauthoryear{Thompson}{1962}]{thompson1962}
			Thompson, K. H. (1962).
			Estimation of the proportion of vectors in a natural population of insects,
			{\em Biometrics} \textbf{18} 568--578.
			
			    \bibitem[Tu et al.(1995)]{tu1995}
			Tu, X. M., Litvak, E., and Pagano, M. (1995).
			On the informativeness and accuracy of pooled testing in estimating prevalence of a rare disease: application to HIV screening.
			\textit{Biometrika} \textbf{82} 287--297.

			\bibitem[Warasi et al.(2016)]{warasi2016}
			Warasi, M. S., Tebbs, J. M., McMahan, C. S., and Bilder, C. R. (2016).
			Estimating the prevalence of multiple diseases from two-stage hierarchical pooling.
			{\em Statistics In Medicine} \textbf{35} 3851--3864.

			\bibitem[Wu(1983)]{wu1983}
                        Wu, C. F. (1983).
                        On the convergence properties of the EM algorithm.
			{\em The Annals of Statistics} \textbf{11} 95--103.
			
			    \bibitem[Zhang et al.(2014)]{zhang2014}
			Zhang, Z., Liu, C., Kim, S., and Liu, A. (2014).
			Prevalence estimation subject to misclassification: the mis-substitution bias and some remedies.
			\textit{Statistics in Medicine} \textbf{33} 4482--4500.

			\end{thebibliography}
\end{document}